\documentclass{amsart}

\usepackage{amsmath, mathrsfs, amssymb,amsthm,hyperref,nameref,centernot,comment,stmaryrd,enumerate,cancel,mathtools}

\hypersetup{pdfstartview={XYZ null null 1.25}}
\usepackage[all]{xy}

\newcommand{\ra}{\rightarrow}
\newcommand{\lra}{\leftrightarrow}

\newcommand{\bv}{\bigvee}
\newcommand{\bw}{\bigwedge}
\newcommand{\GXY}{G(X,Y,\R)}

\newcommand{\cF}{\mathcal F}
\newcommand{\cI}{\mathcal I}

\newcommand{\Pe}{(e_X,e_Y,\R)}

\newcommand{\PXY}{(X,Y,\R)}
\newcommand{\XUY}{X\cup_\preceq Y}
\newcommand{\XUpY}{X\uplus_\preceq Y}
\newcommand{\XUgY}{X\uplus_{\hRg} Y}

\newcommand{\MX}{\overline{X}}
\newcommand{\MY}{\overline{Y}}

\newcommand{\DM}{\mathcal{N}}

\newcommand{\cN}{\mathcal N}
\newcommand{\pol}{\mathrm{Pol}}
\newcommand{\del}{\mathrm{Del}}

\DeclareMathOperator{\Rz}{\preceq_0}
\DeclareMathOperator{\hRm}{\preceq_m}
\DeclareMathOperator{\hRg}{\preceq_g}
\DeclareMathOperator{\nothRg}{\not\preceq_g}

\DeclareMathOperator{\bR}{\overline{R}}
\DeclareMathOperator{\R}{R}

\DeclareMathOperator{\oS}{S}
\DeclareMathOperator{\uS}{\underline{S}}

\newcommand{\bbQ}{\mathbb Q}
\newcommand{\bbR}{\mathbb R}
\newcommand{\id}{\mathrm{id}}
\DeclareMathOperator{\amp}{\&}

\makeatletter
\let\orgdescriptionlabel\descriptionlabel
\renewcommand*{\descriptionlabel}[1]{%
  \let\orglabel\label
  \let\label\@gobble
  \phantomsection
  \edef\@currentlabel{#1}%
  \let\label\orglabel
  \orgdescriptionlabel{#1}%
}
\makeatother

\theoremstyle{plain}
\newtheorem{thm}{Theorem}[section]
\newtheorem{prop}[thm]{Proposition}
\newtheorem{cor}[thm]{Corollary}
\newtheorem{lemma}[thm]{Lemma}
\newtheorem{ex}[thm]{Example}
\theoremstyle{definition}
\newtheorem{defn}[thm]{Definition}

\title{Order polarities}
\author{Rob Egrot}
\date{}
\begin{document}
\begin{abstract}
We define an \emph{order polarity} to be a polarity $(X,Y,\R)$ where $X$ and $Y$ are partially ordered, and we define an \emph{extension polarity} to be a triple $\Pe$ such that $e_X:P\to X$ and $e_Y:P\to Y$ are poset extensions and $(X,Y,\R)$ is an order polarity. We define a hierarchy of increasingly strong coherence conditions for extension polarities, each equivalent to the existence of a preorder structure on $X\cup Y$ such that the natural embeddings, $\iota_X$ and $\iota_Y$, of $X$ and $Y$, respectively, into $X\cup Y$ preserve the order structures of $X$ and $Y$ in increasingly strict ways. We define a Galois polarity to be an extension polarity satisfying the strongest of these coherence conditions, and where $e_X$ and $e_Y$ are meet- and join-extensions respectively. We show that for such polarities the corresponding preorder on $X\cup Y$ is unique.  We define morphisms for polarities, providing the class of Galois polarities with the structure of a category, and we define an adjunction between this category and the category of $\Delta_1$-completions and appropriate homomorphisms.     
\end{abstract}

\subjclass[2010]{Primary 03G10, 06B23 Secondary 18A40, 	06C25, 68T30}

\keywords{Polarity, canonical extension, completion, intermediate structure}

\maketitle

\section{Introduction} 
\subsection{Background}
The concept of a \emph{polarity}, i.e. a pair of sets $X$ and $Y$ and a relation $\R$ between them, was known to Birkhoff at least as far back as 1940 \cite{Bir40}. While, according to \cite[p122]{Bir95}, originally defined as a generalization of the dual isomorphism between polars in analytic geometry, the generality of the definition has lent itself to diverse applications in mathematics and computer science. For example, polarities under the name of \emph{formal concepts} are fundamental in formal concept analysis \cite{GanWil99}. As another example, polarities appear bearing the name \emph{classification} in the theory of information classification \cite[Lecture 4]{BarSel97}, where they are again a foundational concept.  

For a more purely mathematical application, a particular kind of polarity, referred to as a \emph{polarization}, was used in \cite{Tun74} to produce poset completions. The same paper also proves various results connecting properties of polarizations with properties of the resulting completion. More recently, this technique has been exploited to construct canonical extensions for bounded lattice expansions \cite{GehHar01}, and also for posets \cite{DGP05}, where they provide a tool for `completeness via canonicity' results for substructural logics. Something similar also appears implicitly in \cite{GhiMel97}, though neither polarizations nor polarities in general are mentioned explicitly.

The general idea behind these completeness results is, given a poset $P$ equipped with additional operations that are either order preserving or reversing in each coordinate, to show that there exists a completion of $P$ to which the additional operations can be extended. The roots of this technique appear in \cite{JonTar51}, though not in the context of `completeness via canonicity' results, as a generalization of Stone's representation theorem \cite{Stone36} to Boolean algebras with operators (BAOs). The approach there was to first (non-constructively) dualize to relational structures, then construct the canonical extension from these. 

Early generalizations to distributive lattices used Priestly duality \cite{Pr70,Pri72} in a similar way (see for example \cite{GehJon94,S-S00a,S-S00b}). More recent approaches using polarities bypass the dual construction, which is significantly more complicated outside of the distributive setting, and have the additional advantage of being constructive \cite{GehHar01,DGP05}. Indeed, an innovation of \cite{DGP05} is to use the canonical extension of a poset to \emph{construct} a dual, which can then play the same role in providing completeness results for substructural logics as the canonical frame does in the modal setting (see e.g. \cite[Chapters 4 and 5]{BRV01}). For more on the development of the theory of canonical extensions see, for example,  \cite{GehVos11} or the introduction to \cite{Gol18}.

We note that for operations that are not \emph{operators} in the sense of \cite{JonTar51}, the canonical extension construction is ambiguous, as there are often several non-equivalent choices for the lifts of each operation, each of which may be `correct' depending on the situation (see for example the epilogue of \cite{GehJon04} for a brief discussion of this). Moreover, for posets, what is meant by \emph{the} canonical extension is even less clear than it is in the lattice case. This is a consequence of ambiguities surrounding the notions of `filters' and `ideals' in the more general setting. See \cite{Mor14} for a thorough investigation of this issue. 

For canonical extensions in their various guises to play a role in `completeness via canonicity' arguments, general results concerning the preservation of equations and inequalities are extremely useful. Some results of this sort can be found in \cite{Suz11a,Suz11b}, where arguments from \cite{GhiMel97} are extended to more general settings. One component of these arguments is the exploitation of the so called \emph{intermediate structure}, an extension of the original poset intermediate between it and the canonical extension. The idea is that operations are, in a sense, lifted first to the intermediate structure, and then to the canonical extension.    

More generally, the class of $\Delta_1$-completions \cite{GJP13} includes canonical extensions (however we define them), and also others such as the MacNeille (aka \emph{Normal}) completion. Given a poset $P$, the $\Delta_1$-completions of $P$ are, modulo suitable concepts of isomorphism, in one-to-one correspondence with certain kinds of polarities constructed from $P$ \cite[Theorem 3.4]{GJP13}. Here also the intermediate structure appears. Indeed, a $\Delta_1$-completion is the MacNeille completion of its intermediate structure \cite[Section 3]{GJP13}. 

\subsection{What is done here}
In the existing literature, the intermediate structure emerges almost coincidentally from the construction of a completion. Given a polarity $(X, Y, \R)$, first a complete lattice $\GXY$ is constructed using the antitone Galois connection between $\wp(X)$ and $\wp(Y)$ induced by $\R$, as we explain in more detail in Section \ref{S:pol1}. The intermediate structure is then found sitting inside it as a subposet. There are natural maps from $X$ and $Y$ into the intermediate structure, and, if these are injective, partial orderings are thus induced on $X$ and $Y$.  When $\GXY$ is an extension of a poset $P$, it will also follow that $X$ and $Y$ are extensions of $P$. It turns out that the preorder on $X\cup Y$ induced by the intermediate structure agrees with $\R$ on $X\times Y$. 

The broad goal of this paper is to take the idea of a polarity involving order extensions $e_X:P\to X$ and $e_Y:P\to Y$ as primitive, and develop a theory from this. More explicitly, we are interested in the interaction between the relation $\R$ and the orders on $X$ and $Y$, and, in particular, under what circumstances something corresponding to the `intermediate structure' can be defined on $X\cup Y$. This issue raises several questions, depending on exactly what properties we think an `intermediate structure' should have.

Based on our answers to these questions, we define a sequence of so-called \emph{coherence conditions} for polarities. The choice of name here comes from the idea that we have three things, $\leq_X$, $\leq_Y$ and $\R$, all giving us information about a possible preorder on $X\cup Y$ with certain properties, and we want this information to not contradict itself. The bulk of this work is done in Section \ref{S:order}, where the main definitions are made, and in Section \ref{S:satisfaction}, where, among other things, we prove our defined conditions are strictly increasing in strength. 

In Section \ref{S:Galois} we define a \emph{Galois polarity} to be a triple $\Pe$ satisfying the strongest of our coherence conditions, and with the additional property that $e_X$ is a meet-extension, and $e_Y$ is a join-extension. The `aptness' of this definition is partly demonstrated by the fact that, if $\Pe$ is a Galois polarity, there is one and only one possible preorder structure on $X\cup Y$ that agrees with the orders on $X$ and $Y$, agrees with $\R$ on $X\times Y$, and also preserves meets and joins from the base poset $P$ (see Theorem \ref{T:unique} for a more precise statement).

Galois polarities are studied further in Section \ref{S:GaloiS1}. First we justify the choice of terminology by demonstrating that, for Galois polarities, the unique preorder structure described above can be defined in terms of a Galois connection between any join-preserving join-completion of $Y$ and any meet-preserving meet-completion of $X$. This requires some technical results on extending and restricting polarity relations, which we provide in Section \ref{S:ExtRes}. Here we investigate the `simplest' way we might hope to extend a relation between posets to a relation between meet- and join-extensions of these posets, and conversely the simplest way we might restrict a relation between extensions to a relation between the original posets. In particular we prove that it is rather common for coherence properties of a polarity to be preserved by extension and restriction as we define them.    

By defining suitable morphisms, we can equip the class of Galois polarities with the structure of a category. This can be seen as a generalization of the concept of a $\delta$-homomorphism from \cite[Section 4]{GehPri08}. We define an adjunction between this category and the category of $\Delta_1$-completions (see Theorem \ref{T:adj}). This produces the correspondence between $\Delta_1$-completions of a poset and certain kinds of polarities from \cite[Theorem 3.4]{GJP13} via the categorical equivalence between fixed subcategories.

In the long term we imagine handling lifting of operations, and the preservation of inequalities and so on, to `intermediate structures' induced by polarities, and Galois polarities in particular. This is, of course, not an entirely new idea. Indeed, we have mentioned previously that lifting operations to canonical extensions is often done by first lifting to the intermediate structure. The hope is that, by shifting the focus a little from intermediate structures as they emerge in the construction of completions, to intermediate structures as algebraic objects of interest in their own right, some new insight might be gained. However, to control the length of this document, we leave the pursuit of this rather vague goal to future work.

\section{Orders and completions}\label{S:comp}
\subsection{A note on notation}
We use the following not entirely standard notations:
\begin{itemize}
\item Given a poset $P$ and $p\in P$, we define 
\[p^\uparrow = \{q\in P: q\geq p\} \text{ and }p^\downarrow= \{q\in P: q\leq p\}.\]
\item Given a function $f:X\to Y$, and given $S\subseteq X$, we define 
\[f[S]=\{f(x):x\in S\}.\]
\item With $f$ as above and with $y\in Y$ and $T\subseteq Y$ we define 
\[f^{-1}(y) = \{x\in X: f(x)=y\}\] and \[f^{-1}(T)=\{x\in X:f(x)\in T\}.\]
\item If $P$ is a poset, then $P^\partial$ is the order dual of $P$.
\item If $X$ and $Y$ are sets, then we may refer to a relation $\R\subseteq X\times Y$ as being a \emph{relation on $X\times Y$}.
\end{itemize}

\subsection{Extensions and completions}
We assume familiarity with the basics of order theory. Textbook exposition can be found in \cite{DavPri02}. In this subsection we provide a brisk introduction to some more advanced order theory concepts. This serves primarily to establish the notation we will be using.  

\begin{defn}Let $P$ and $Q$ be posets. We say an order embedding $e:P\to Q$ is a \textbf{poset extension}, or just an \emph{extension}. If $Q$ is also a complete lattice we say $e$ is a \textbf{completion}. If for all $q\in Q$ we have $q = \bw e[e^{-1}(q^\uparrow)]$, then we say $e$ is a \textbf{meet-extension}, or a \textbf{meet-completion} if $Q$ is a complete lattice. Similarly, if $q = \bv e[e^{-1}(q^\downarrow)]$ for all $q\in Q$, then $e$ is a \textbf{join-extension}, or a \textbf{join-completion} when $Q$ is complete.
\end{defn}

Note that it is common in the literature to refer to completions using the codomain of the function. For example, we might say ``$Q$ is a completion of $P$" when talking about the completion $e:P\to Q$. This has the disadvantage of obfuscating the issue of what it means for two extensions to be isomorphic, as an isomorphism between codomains is not sufficient for extensions to be isomorphic in the sense used here (see below). This rarely causes significant problems in practice, as it is usually clear from context what kind of isomorphism is required. However, we find the identification of extensions with maps to be more elegant, and will generally use this approach.    

\begin{defn}[Morphisms between order preserving maps and extensions]\label{D:mapCat}
Given posets $P_1, P_2, Q_1, Q_2$, and  order preserving maps $f_1:P_1\to Q_1$ and $f_2:P_2\to Q_2$, a map, or morphism, from $f_1$ to $f_2$ is a pair of order preserving maps $g_P:P_1\to P_2$ and $g_Q:Q_1\to Q_2$ such that the diagram in Figure \ref{F:mapHom} commutes. If $g_p$ and $g_Q$ are both order isomorphisms, then we say $f_1$ and $f_2$ are isomorphic. 

If $f_1:P\to Q_1$ and $f_2:P\to Q_2$ are extensions of a poset $P$, then $f_1$ and $f_2$ are \textbf{isomorphic as extensions of $P$} if they are isomorphic in the sense described above and the map $g_P$ is the identity on $P$.
\end{defn}

Definition \ref{D:mapCat} equips the class of order preserving maps between posets, and in particular the subclass of poset extensions, with the structure of a category. We will make frequent use of the idea of extensions being isomorphic, and we will return to the idea of a category of extensions in Section \ref{S:cat}. 

\begin{figure}[htbp]
\[\xymatrix{ P_1\ar[r]^{f_1}\ar[d]_{g_P} & Q_1\ar[d]^{g_Q} \\
P_2\ar[r]_{f_2} & Q_2
}\] 
\caption{}
\label{F:mapHom}
\end{figure}

\begin{defn}\label{D:DM}
Given a poset $P$, the \textbf{MacNeille completion} of $P$ is a map $e:P\to \cN(P)$ that is both a meet- and a join-completion.
\end{defn}

The MacNeille completion was introduced in \cite{Mac37} as a generalization of Dedekind's construction of $\bbR$ from $\bbQ$. It is unique up to isomorphism. The characterization used here is due to \cite{BanBru67}. See e.g. \cite[Section 7.38]{DavPri02} for more information. Note that a MacNeille completion $e$ is completely join- and meet-preserving.

\begin{defn}\label{D:can}
The \textbf{canonical extension} of a lattice $L$ is a completion $e:L\to L^\delta$ such that:
\begin{enumerate} 
\item $e[L]$ is \emph{dense} in $L^\delta$. I.e. Every element of $L^\delta$ is expressible both as a join of meets, and as a meet of joins, of elements of $e[L]$.
\item $e$ is \emph{compact}. I.e. for all $S, T\subseteq L$, if $\bw e[S]\leq \bv e[T]$, then there are finite $S'\subseteq S$ and $T'\subseteq T$ with $\bw S'\leq \bv T'$. 
\end{enumerate}
\end{defn}

Canonical extensions are also unique up to isomorphism. This characterization, and the proof that such a completion exists for all $L$, is due to \cite{GehHar01}. It generalizes the definitions of the canonical extension for Boolean algebras \cite{JonTar51}, and distributive lattices \cite{GehJon94}. The construction used in \cite{GehHar01} can, as noted in Remark 2.8 of that paper, also be used for posets, and will again result in a dense completion. However, the kind of compactness obtained is weaker. This idea is expanded upon in \cite{DGP05}. The differences between the lattice and poset cases arise from the fact that definitions for filters and ideals which are equivalent for lattices are not so for posets. This issue is discussed in detail in \cite{Mor14}. One way to address this systematically is to talk about \emph{the canonical extension of $P$ with respect to $\cF$ and $\cI$}, where $\cF$ and $\cI$ are sets of `filters' and `ideals' of $P$ respectively. By making the definitions of `filter' and `ideal' weak enough, this allows all notions of the canonical extension of a poset to be treated in a uniform fashion. This is the approach taken in \cite{MorVanA18}, for example.   

\begin{defn}\label{D:Del}
Given a poset $P$, a $\Delta_1$\textbf{-completion} of $P$ is a completion $e:P\to D$ such that $e[P]$ is dense in $D$.
\end{defn}

$\Delta_1$-completions, introduced in \cite{GJP13}, include both MacNeille completions and canonical extensions. As such they are not usually unique up to isomorphism, so it doesn't make sense to talk about \emph{the} $\Delta_1$-completion.  

\begin{defn}
Let $P$ and $Q$ be posets. Then a \textbf{monotone Galois connection}, or just a \emph{Galois connection}, between $P$ and $Q$ is a pair of order preserving maps $\alpha:P\to Q$ and $\beta:Q\to P$ such that, for all $p\in P$ and $q\in Q$, we have
\[\alpha(p)\leq q \iff p\leq \beta(q).\]
The map $\alpha$ is the \textbf{left adjoint}, and $\beta$ is the \textbf{right adjoint}. 

An \textbf{antitone Galois connection} between $P$ and $Q$ is a Galois connection between $P$ and the order dual, $Q^\partial$, of $Q$.
\end{defn}

\begin{defn}
A \textbf{preorder} on a set is a binary relation that is reflexive and transitive. Every preorder induces a \textbf{canonical partial order} by identifying pairs of elements that break anti-symmetry.
\end{defn}

\subsection{Polarities for completions}\label{S:pol1}

Following \cite{Bir40}, we define a \textbf{polarity} to be a triple $\PXY$, where $X$ and $Y$ are sets, and $\R\subseteq X\times Y$ is a binary relation. For convenience we will assume also that $X$ and $Y$ are disjoint. See the section on polarities in \cite{EKMS93} for several examples. Polarities have also been called \emph{polarity frames} \cite{Suz14}. Given any polarity $\PXY$, there is an antitone Galois connection between $\wp(X)$ and $\wp(Y)$. This is given by the order reversing maps $(-)^R:\wp(X)\to\wp(Y)$ and ${}^R(-):\wp(Y)\to\wp(X)$ defined as follows:
\[(S)^R = \{y\in Y: x\R y \text{ for all } x\in S\}.\]
\[{}^R(T) = \{x\in X: x\R y \text{ for all } y\in T\}.\]

The set $\GXY$ of subsets of $X$ that are fixed by the composite map ${}^R(-)\circ (-)^R$ is a complete lattice. Indeed, this is a closure operator on $\wp(X)$. 

Polarities in the special case where $X$ and $Y$ are sets of subsets of some common set $Z$, where the relation $\R$ is that of non-empty intersection, and which also satisfy some additional conditions, have been referred to as \emph{polarizations} in the literature \cite{Tun74, MorVanA18}. Polarizations play an important role in the construction of canonical extensions. 

There are maps $\Xi:X\to\GXY$ and $\Upsilon:Y\to \GXY$ defined by:
\[\Xi(x) = {}^R(\{x\}^R)\text{ for $x\in X$, and}\]  
\[\Upsilon(y) = {}^R\{y\} \text{ for $y\in Y$}.\]
$\Xi[X]$ and $\Upsilon[Y]$ join- and meet-generate $\GXY$ respectively \cite[Proposition 2.10]{Geh06}. Moreover, the (not usually disjoint) union $\Xi[X]\cup\Upsilon[Y]$ inherits an ordering from $\GXY$. Thus the inclusion of the poset $\Xi[X]\cup\Upsilon[Y]$ into $\GXY$ can be characterized as the MacNeille completion of $\Xi[X]\cup\Upsilon[Y]$. The order on $\Xi[X]\cup\Upsilon[Y]$ can be defined without first constructing $\GXY$. We expand on this in Proposition \ref{P:order} below. 

\begin{prop}\label{P:order}
Define a preorder $\preceq$ on $\Xi[X]\cup\Upsilon[Y]$ as follows: 
\begin{enumerate}
\item $\forall x_1,x_2\in X\Big(\Xi(x_1)\preceq \Xi(x_2)\lra \forall y\in Y\big(x_2 \R  y\ra x_1 \R  y\big)\Big)$.
\item $\forall y_1,y_2\in Y\Big(\Upsilon(y_1)\preceq \Upsilon(y_2)\lra \forall x\in X\big(x\R  y_1\ra x\R y_2\big)\Big)$.
\item $(\forall x\in X)(\forall y\in Y)\Big(\Xi(x)\preceq\Upsilon(y)\lra x\R y\Big)$.
\item \[(\forall x\in X)(\forall y\in Y)\Big(\Upsilon(y)\preceq\Xi(x)\lra(\forall x'\in X)(\forall y'\in Y)\big((x'\R y\amp x\R y')\ra x'\R y'\big)\Big).\] 
\end{enumerate}
The partial ordering of $\Xi[X]\cup\Upsilon[Y]$ inherited from $\GXY$ is the canonical partial order induced by $\preceq$. 
\end{prop}
\begin{proof}
This is essentially \cite[Proposition 2.7]{Geh06}. 
\end{proof}
Proposition \ref{P:nat} below provides another perspective on the conditions from Proposition \ref{P:order}.

\begin{prop}\label{P:nat}
Let $\PXY$ be a polarity. Then the following are equivalent:
\begin{enumerate}[1.] 
\item $\preceq$ is the least preorder definable on $\Xi[X]\cup \Upsilon[Y]$ such that: 
\begin{enumerate}
\item $\Xi(x)\preceq\Upsilon(y)\iff x\R y$ for all $x\in X$ and $y\in Y$.
\item The restrictions of $\preceq$ to $\Xi[X]$ and $\Upsilon[Y]$ agree with the orders on these sets inherited from $\GXY$. 
\end{enumerate}
\item $\preceq$ satisfies the conditions from Proposition \ref{P:order}
\end{enumerate}
\end{prop} 
\begin{proof}
Suppose $\preceq$ is any preorder on $\Xi[X]\cup \Upsilon[Y]$ satisfying conditions 1(a) and 1(b). Then, by Proposition \ref{P:order} we have 
\[\Xi(x_1)\preceq \Xi(x_2)\iff \Xi(x_1)\subseteq \Xi(x_2)\iff (\forall y\in Y)\big(x_2 \R  y\ra x_1 \R  y\big),\] and thus \ref{P:order}(1) is satisfied. A similar argument works for \ref{P:order}(2), and \ref{P:order}(3) holds automatically. Finally, as $\preceq$ is transitive, we must have 
\begin{align*}&\Upsilon(y) \preceq \Xi(x)\\\implies& (\forall x'\in X)(\forall y'\in Y)\Big(\Xi(x')\preceq \Upsilon(y)\amp\Xi(x)\preceq \Upsilon(y')\ra \Xi(x')\preceq \Upsilon(y')\Big).\end{align*}
Thus, any such preorder $\preceq$ satisfies \ref{P:order}(1)-(3), and the `forward implication only' version of \ref{P:order}(4). 

To complete the proof it is sufficient to show that the `minimal' $\preceq$ defined from $\R $ using conditions \ref{P:order}(1)-(4) defines a preorder on $\Xi[X]\cup \Upsilon[Y]$ satisfying conditions 1(a) and 1(b). But this is what Proposition \ref{P:order} tells us. 
\end{proof}

\begin{lemma}\label{L:inj}
The following are equivalent:
\begin{enumerate}
\item[(1.a)] The map $\Xi:X\to \GXY$ is injective.
\item[(1.b)] Whenever $x_1\neq x_2\in X$ there is $y\in Y$ such that either $(x_2,y)\in \R $ and $(x_1, y)\notin \R$, or vice versa. 
\item[(1.c)] Whenever $x_1\neq x_2\in X$ we have either $x_1\notin \Xi(x_2)$ or $x_2\notin \Xi(x_1)$. 
\end{enumerate}
The following are also equivalent:
\begin{enumerate}
\item[(2.a)] The map $\Upsilon:Y\to\GXY$ is injective.
\item[(2.b)] Whenever $y_1\neq y_2\in Y$ there is $x\in X$ such that either $(x,y_2)\in \R$ and $(x,y_1)\notin \R$, or vice versa. 
\end{enumerate}
\end{lemma}
\begin{proof}
Observe that $\Xi(x) = \{z\in X:  x\R y\ra z\R y\text{ for all }y\in Y\}$ for all $x\in X$. Let $x_1\neq x_2$ and suppose without loss of generality that there is $z\in \Xi(x_1)\setminus\Xi(x_2)$. Then $(z,y)\in\R$ for all $y\in Y$ with $(x_1,y)\in \R$, but there is $y'\in Y$ with $(x_2,y')\in\R$ and $(z,y')\notin \R$. For this $y'$ we must have $(x_2,y')\in \R$ and $(x_1,y')\notin \R$. Thus $(1.a)\implies(1.b)$. That $(1.b)\implies(1.c)$ and $(1.c)\implies(1.a)$ is automatic. The proof for $\Upsilon$ is similar, but even more straightforward.
\end{proof}

What if $X$ and $Y$ are not merely sets but also have a poset structure? We make the following definition.
\begin{defn}\label{D:OP}
A polarity $\PXY$ is an \textbf{order polarity} if $X$ and $Y$ are posets.
\end{defn}

In this situation we might, for example, want the maps $\Xi$ and $\Upsilon$ to be order embeddings, which places constraints on $\R $. Building on Lemma \ref{L:inj} we have the following result.

\begin{prop}\label{P:ordEmb}
Let $\PXY$ be an order polarity. Then the map $\Xi:X\to \GXY$ is an order embedding if and only if
\[\forall x_1,x_2\in X \Big(x_1\leq_X x_2\lra \forall y\in Y\big(x_2 \R  y\ra x_1 \R  y\big)\Big).\] 

The map $\Upsilon:Y\to \GXY$ is an order embedding if and only if
\[\forall y_1,y_2\in Y \Big(y_1\leq_Y y_2\lra \forall x\in X\big(x\R  y_1\ra x \R  y_2\big)\Big).\] 
\end{prop}
\begin{proof}
We could appeal to Proposition \ref{P:order}, but the direct argument is also extremely simple. Explicitly, $\Xi$ is an order embedding if and only if $x_1\leq_X x_2 \iff \Xi(x_1)\subseteq \Xi(x_2)$, and a little consideration reveals that $\Xi(x_1)\subseteq \Xi(x_2)$ if and only if $x_2 \R  y\ra x_1 \R  y$ for all $y\in Y$. Again, the argument for $\Upsilon$ is even more straightforward. 
\end{proof}

Propositions \ref{P:nat} and \ref{P:ordEmb}, while essentially trivial in themselves, contain, in a sense, the seed of inspiration for the rest of the paper. In broad terms, we want to investigate the conditions for the existence of preorders on $X\cup Y$ such that similar results can be proved. This we do in the next section and onwards. First, a little more notation.

\begin{defn}[$\XUY$, $\XUpY$]\label{D:XUP}
Given disjoint sets $X$ and $Y$, we sometimes write $\XUY$ to specify that we are talking about $X\cup Y$ ordered by a given preorder $\preceq$. We use $\XUpY$ to denote the canonical partial order arising from $\XUY$.
\end{defn}

\section{Coherence conditions for order polarities}\label{S:order}
\subsection{The basic case}
In the previous section we discussed polarities and order polarities from the perspective of $\GXY$, and the inherited order structure on $\Xi[X]\cup\Upsilon[Y]$. In this situation the maps $\Xi$ and $\Upsilon$ may fail to be order preserving, order reflecting, or even injective. In this section we forget about $\GXY$, and ask instead, given an order polarity $\PXY$, under what circumstances can we define preorders on $X\cup Y$ that agree with $\R $ on $X\times Y$, and also extend the order structures of $X$ and $Y$? In other words, when are there preorders on $X\cup Y$ agreeing with $\R $ on $X\times Y$ such that the natural inclusions of $X$ and $Y$ into $X\cup Y$ are order preserving? What about if we require the inclusions to be order embeddings, or to have stronger preservation properties? We will address these questions, but first some definitions.

\begin{defn}\label{D:0-pre}
Let $\PXY$ be an order polarity. Define a \textbf{0-preorder} for $\PXY$ to be a preorder $\preceq$ on $X\cup Y$ with additional properties as follows:
\begin{description}
\item[(P1)\label{P1}] $(\forall x\in X)(\forall y\in Y)\big( x\preceq y \lra x\R y\big)$.
\item[(P2)\label{P2}] $\forall x_1,x_2\in X(x_1\leq_X x_2\ra x_1\preceq x_2)$.
\item[(P3)\label{P3}] $\forall y_1,y_2\in Y(y_1\leq_Y y_2\ra y_1\preceq y_2)$.
\end{description}
\end{defn}

\begin{defn}\label{D:0-co}
Let $\PXY$ be an order polarity. We say $\PXY$ is \textbf{0-coherent} if it satisfies the following conditions:
\begin{description}
\item[(C1)\label{C1}] $(\forall x_1,x_2\in X)(\forall y\in Y)\Big(x_1 \leq_X x_2 \ra (x_2 \R  y\ra x_1 \R  y)\Big)$. 
\item[(C2)\label{C2}] $(\forall y_1,y_2\in Y)(\forall x\in X) \Big(y_1\leq_Y y_2\ra (x\R  y_1\ra x\R y_2)\Big)$. 
\end{description}
\end{defn}

\begin{defn}[$\Rz$]\label{D:Rz}
Let $O=\PXY$ be an order polarity. Define the relation $\Rz\subseteq(X\cup Y)^2$ by 
\[\Rz=\leq_X\cup \leq_Y\cup \R.\]
\end{defn}

Technically, $\Rz$ depends on the choice of $O$, so one could make the case that the notation should be something like $\preceq_0^O$ to reflect that. We choose not do this, as we find it rather unwieldy, and the choice of $O$ will always be obvious from the context. Similar issues arise later at several points, for example in Definitions \ref{D:hRm} and \ref{D:hRg}, and we take the same approach. 

\begin{thm}\label{T:0-pol}
Given an order polarity $O=\PXY$, the following are equivalent:
\begin{enumerate}
\item $O$ is 0-coherent.
\item $\Rz$ is a 0-preorder for $O$.
\item There exists a 0-preorder for $O$.
\end{enumerate}
Moreover, the set of 0-preorders for $\PXY$ is closed under non-empty intersections and, if it is non-empty, has $\Rz$ as its smallest member. 
\end{thm}
\begin{proof}
Suppose first that $\PXY$ is 0-coherent. It is immediate from its definition that $\Rz$ is reflexive, so it remains only to check transitivity. To do this we consider triples $(z_1,z_2,z_3) \in (X\cup Y)^3$, with $z_1 \Rz z_2$, and $z_2\Rz z_3$. A simple counting argument reveals there are eight cases, depending on the containment of each $z_i$ in $X$ or $Y$. The cases where the $z$ values are either all in $X$ or all in $Y$ follow from the fact that $\Rz$ agrees with the orders on $X$ and $Y$. The cases that require $y\Rz x$ are ruled out by the definition of $\Rz$, so the only remaining cases are $(x_1, x_2, y)$, where $x_1,x_2\in X$ and $y\in Y$, and $(x, y_1, y_2)$ where $x\in X$ and $y_1,y_2\in Y$. These cases are covered by the definition of 0-coherence, so we have transitivity, and thus $(1)\implies (2)$.

That $(2)\implies (3)$ is automatic, so suppose now that $\preceq$ is a 0-preorder for $\PXY$. Then that $\PXY$ is 0-coherent is an almost immediate consequence of the transitivity of $\preceq$ along with the additional conditions from Definition \ref{D:0-pre}. Thus $(3)\implies (1)$. 

Finally, the set of 0-preorders for $\PXY$ is obviously closed under non-empty intersections. Moreover, $\Rz$ must be a member of this set (if it is non-empty), by what we have proved already. The proof that it must be the smallest member is routine.
\end{proof}

Conditions \ref{C1} and \ref{C2} are, in a sense, dual to each other, as are several other pairs of conditions to come. We will often appeal to this duality in proofs. Informally, we mean something like ``by switching some conditions to their (intuitively obvious) duals we could prove this using essentially the same argument". It is possible to formalize this intuition, but we omit the details for reasons of space. The ad hoc understanding suffices to reconstruct proofs as necessary. 
  
\subsection{Extension polarities}

Suppose in addition that $X$ and $Y$ are both extensions of some poset $P$. In other words, that there are order embeddings $e_1:P\to X$ and $e_2:P\to Y$. What conditions must $\R $ satisfy in order for there to be a 0-preorder $\preceq$ such that the diagram in Figure \ref{F:fix1} commutes (recall Definition \ref{D:XUP})? In this figure $\iota_X$ and $\iota_Y$ stand for the compositions of the natural inclusion functions into $\XUY$ with the canonical map from $\XUY$ to $\XUpY$. As this situation will be the focus of most of the rest of the document, we make the following definition. 

\begin{figure}[htbp]
\[\xymatrix{ P\ar[r]^{e_Y}\ar[d]_{e_X} & Y\ar[d]^{\iota_Y} \\
X\ar[r]_{\iota_X} & \XUpY
}\] 
\caption{}
\label{F:fix1}
\end{figure}

\begin{defn}\label{D:ext}
An \textbf{extension polarity} is a triple $(e_X,e_Y,\R )$, where $e_X:P \to X$ and $e_Y:P\to Y$ are order extensions of the same poset $P$, and $(X,Y,\R)$ is an order polarity. When both $e_X$ and $e_Y$ are completions, we say $(e_X,e_Y,\R )$ is \textbf{complete}. We sometimes say an extension polarity $\Pe$ \textbf{extends} $P$. The concepts of 0-preorders and 0-coherence, from Definition \ref{D:0-pre} and \ref{D:0-co}, also apply, \emph{mutatis mutandis}, to extension polarities. 
\end{defn}

Note that an order polarity is an extension polarity in the special case where $P$ is empty. 

\begin{defn}\label{D:1-pre}
Let $E=\Pe$ be an extension polarity. Define a \textbf{1-preorder} for $E$ to be a 0-preorder for $E$ with the additional property that the diagram in Figure \ref{F:fix1} commutes.
\end{defn}

\begin{defn}\label{1-co}
Let $E=\Pe$ be an extension polarity. We say $E$ is \textbf{1-coherent} if it is 0-coherent and also satisfies the following conditions:
\begin{description}
\item[(C3)\label{C3}] $\forall p\in P\Big(e_X(p) \R  e_Y(p)\Big)$. 
\item[(C4)\label{C4}] $(\forall p \in P)(\forall x\in X)(\forall y\in Y) \Big((x \R   e_Y(p) \amp e_X(p) \R  y)\ra x\R  y\Big)$. 
\end{description} 
\end{defn}

\begin{defn}[$\hRm$]\label{D:hRm}
Let $\Pe$ be an extension polarity. Define the relation $\hRm\subseteq (X\cup Y)^2$ to be the union of $\R$ with the sets
\begin{align*}
Z_X =&\{(x_1,x_2)\in X^2:\exists p\in P\big(x_1 \R e_Y(p) \amp e_X(p)\leq_X x_2\big)\}\cup \leq_X \\
Z_Y=& \{(y_1,y_2)\in Y^2:\exists p\in P\big(y_1\leq_Y e_Y(p) \amp e_X(p)\R y_2\big)\}\cup \leq_Y \\
Z_{YX} =& \{(y,x)\in Y\times X : \exists p,q\in P\big(y\leq_Y e_Y(p) \amp e_X(p)\R e_Y(q) \amp e_X(q)\leq_X x\big)\}.
\end{align*}
\end{defn}

\begin{lemma}\label{L:1-co-equiv}
Let $E=\Pe$ be a 1-coherent extension polarity, and let $\preceq$ be a 0-preorder for $E$. Then the following conditions are equivalent:

\[\tag{1}\forall p\in P\Big(e_Y(p) \preceq e_X(p)\Big).\]
\[\tag{2} Z_{YX}\subseteq \preceq.\]

\end{lemma}
\begin{proof}
Assume (1) holds and that there are $p,q\in P$ such that we have $y\leq_Y e_Y(p)$, and $e_X(p) \R  e_Y(q)$, and $e_X(q)\leq_X x$. Then
\[y\preceq e_Y(p)\preceq e_X(p)\preceq e_Y(q)\preceq e_X(q)\preceq x,\]
and so $y\preceq x$ by transitivity, and thus (1)$\implies$ (2). Conversely, if we assume $Z_{YX}\subseteq \preceq$, then setting $y = e_Y(p)$ and $x = e_X(p)$ produces $e_Y(p)\preceq e_X(p)$, and thus (1) and (2) are equivalent as claimed. 
\end{proof}

\begin{lemma}\label{L:hRm-is-1-pre}
Let $E=\Pe$ be a 1-coherent extension polarity, and let $\preceq$ be a 0-preorder for $E$. Then $\preceq$ is a 1-preorder for $E$ if and only if $\preceq$ satisfies the conditions from Lemma \ref{L:1-co-equiv}. Moreover, if these conditions are satisfied, then $Z_X\cup Z_Y\subseteq \preceq$, where $Z_X$ and $Z_Y$ are as in Definition \ref{D:hRm}.
\end{lemma}
\begin{proof}
It is clearly necessary that (1) hold in order for $\preceq$ to be a 1-preorder, as otherwise the diagram in Figure \ref{F:fix1} will not commute. Conversely, if $\preceq$ satisfies (1), then, from \ref{C3} it follows that the diagram in Figure \ref{F:fix1} commutes, and thus $\preceq$ is a 1-preorder. 

Now, suppose the conditions are satisfied, and let $x_1\leq x_2\in X$. Then $x_1\preceq x_2$ by \ref{P2}. This shows that $\leq_X\subseteq \preceq$.   Suppose then that $x_1 \R  e_Y(p)$ and $e_X(p)\leq_X x_2$ for some $x_1,x_2\in X$ and $p\in P$. Then $x_1\preceq e_Y(p)\preceq e_X(p)\preceq x_2$, and so we must have $x_1\preceq x_2$ by transitivity. It follows that $Z_X\subseteq \preceq$, and that $Z_Y\subseteq \preceq$ follows from a dual argument.
\end{proof}

\begin{thm}\label{T:1-pol}
Given an extension polarity $E=\Pe$, the following are equivalent:
\begin{enumerate}
\item $E$ is 1-coherent.
\item $\hRm$ is a 1-preorder for $E$.
\item There exists a 1-preorder for $E$.
\end{enumerate}
Moreover, the set of 1-preorders for $E$ is closed under non-empty intersections and, if it is non-empty, has $\hRm$ as its smallest member. 
\end{thm}
\begin{proof}
Suppose first that $E$ is 1-coherent. That $\hRm$ is reflexive is automatic, so we show now that it is transitive. As in the proof of Theorem \ref{T:0-pol}, we consider the eight relevant cases of the triples $(z_1, z_2, z_3) \in (X\cup Y)^3$. Unfortunately we must proceed case by case, and each case may have several subcases.
\begin{enumerate}[$\bullet$]
\item$(x_1,x_2,x_3)$: Here $x_1\hRm x_2$, and $x_2\hRm x_3$. This case breaks down into subcases, depending on the reason $\hRm$ holds for each pair. 
\begin{enumerate}[$-$]
\item If $x_1 \leq_X x_2$ and $x_2\leq_X x_3$ in $X$, then we have $x_1\leq_X x_3$, and thus $x_1\hRm x_3$, by transitivity of $\leq_X$. 
\item Suppose instead that $x_1\leq_X x_2$, and that there is $p\in P$ with  $x_2 \R  e_Y(p)$ and $e_X(p)\leq_X x_3$. Then $x_1 \R  e_Y(p)$ by \ref{C1}, and so $x_1\hRm x_3$ by definition of $\hRm$.
\item Alternatively, if $x_1 \R  e_Y(p)$, $e_X(p)\leq_X x_2$ and $x_2\leq_X x_3$, then $e_X(p)\leq_X x_3$, and so $x_1\hRm x_3$ by definition of $\hRm$. 
\item Finally, suppose there are $p,q\in P$ with $x_1 \R  e_Y(p)$, with $e_X(p)\leq_X x_2$, with $x_2 \R  e_Y(q)$ and with $e_X(q)\leq_X x_3$. Then $e_X(p) \R  e_Y(q)$ by \ref{C1}, and so $x_1 \R  e_Y(q)$ by \ref{C4}, and thus $x_1\hRm x_3$ by definition of $\hRm$.
\end{enumerate}

\item$(y_1,y_2,y_3)$: This case is dual to the previous one.

\item$(x_1,x_2,y)$: Here we have $x_2 \hRm y$, and thus $x_2 \R  y$. We also have $x_1\hRm x_2$, which breaks down into two cases.
\begin{enumerate}[$-$]
\item First suppose $x_1\leq_X x_2$. Then $x_1 \R  y$ by \ref{C1}, and so $x_1\hRm y$ as required.
\item Suppose instead that there is $p\in P$ with $x_1 \R  e_Y(p)$ and $e_X(p) \leq_X x_2$. Then $e_X(p) \R  y$ by \ref{C1}, and so $x_1 \R  y$ by \ref{C4}, and thus $x_1 \hRm y$ as required.
\end{enumerate}

\item$(y,x_1,x_2)$: Here we have $y\hRm x$, and so there are $p,q\in P$ with $y\leq_Y e_Y(p)$, with $e_X(p)\R e_Y(q)$, and with $e_X(q)\leq_X x_1$, and $x_1\hRm x_2$. There are two subcases. 
\begin{enumerate}[$-$]
\item Suppose first that $x_1\leq_X x_2$. Then $e_X(q)\leq_X x_2$ by the transitivity of $\leq_X$,  and the result then follows immediately from the definition of $\hRm$.
\item Suppose instead that there is $r\in P$ with $x_1 \R  e_Y(r)$ and $e_X(r) \leq_X x_2$. Then an application of \ref{C1} produces $e_X(q) \R  e_Y(r)$. Using this with \ref{C4} provides $e_X(p) \R  e_Y(r)$. Thus we get $y\hRm x_2$ from the definition of $\hRm$. 
\end{enumerate}
\item$(x_1,y,x_2)$: We have $x_1 \R  y$, and, by the definition of $\hRm$, there are $p,q\in P$ with $y\leq_Y e_Y(p)$, with $e_X(p) \R  e_Y(q)$, and with $e_X(q)\leq_X x_2$. Then \ref{C2} gives us $x_1 \R  e_Y(p)$, and consequently \ref{C4} produces $x_1 \R  e_Y(q)$. Thus $x_1\hRm x_2$ by the definition of $\hRm$. 

\item$(y_1, x,y_2)$: Dual to the previous case.

\item$(x,y_1,y_2)$: Dual to the $(x_1,x_2,y)$ case. 
\item$(y_1,y_2,x)$: Dual to the $(y, x_1,x_2)$ case. 
\end{enumerate}

From the above argument we conclude that $\hRm$ is transitive, and thus defines a preorder. To see that $\hRm$ is a 1-preorder first note that it is obviously a 0-preorder (just examine Definitions \ref{D:0-pre} and \ref{D:hRm}). That $\hRm$ is a 1-preorder then follows immediately from Lemma \ref{L:hRm-is-1-pre}, as we have $Z_{YX}\subseteq \hRm$ by definition. Thus (1)$\implies$(2).

That (2)$\implies$(3) is automatic, so suppose (3) holds, and let $\preceq$ be a 1-preorder for $E$. Then \ref{C3} must hold as otherwise the diagram in Figure \ref{F:fix1} would not commute. Similarly, if this diagram commutes, then we must have $e_Y(p)\preceq e_X(p)$ for all $p\in P$. So, given $p\in P$, $x\in X$ and $y\in Y$ with $x \R e_Y(p)$ and $e_X(p) \R y$, we have $x\preceq e_Y(p)\preceq e_X(p)\preceq y$, as $\preceq$ is a 1-preorder, and thus $x\preceq y$ by transitivity of $\preceq$. By the definition of a 1-preorder this implies $x\R y$. So \ref{C4} also holds, and thus (3)$\implies$ (1).

Finally, that the set of 1-preorders for $E$  is closed under non-empty intersections is easily seen. It follows immediately from Lemma \ref{L:hRm-is-1-pre} that, if non-empty, the smallest member is $\hRm$.  
\end{proof}

\begin{defn}\label{D:2-pre}
Let $E=\Pe$ be an extension polarity. Define a \textbf{2-preorder} for $E$ to be a 1-preorder for $E$ with the additional property that the maps $\iota_X$ and $\iota_Y$ from the diagram in Figure \ref{F:fix1} are both order embeddings.
\end{defn}

\begin{defn}\label{2-co}
Let $\Pe$ be an extension polarity. We say $\Pe$ is \textbf{2-coherent} if it is 1-coherent and also satisfies the following conditions:
\begin{description}
\item[(C5)\label{C5}] $(\forall x_1,x_2\in X)(\forall p\in P)\Big( (x_1 \R  e_Y(p)\amp e_X(p)\leq_X x_2) \ra x_1\leq_X x_2\Big)$.
\item[(C6)\label{C6}] $(\forall y_1,y_2\in Y)(\forall p \in P)\Big((y_1\leq_Y e_Y(p)\amp e_X(p) \R  y_2) \ra y_1\leq_Y y_2\Big)$.
\end{description}
\end{defn}

\begin{thm}\label{T:2-pol}
Given an extension polarity $E=\Pe$, the following are equivalent:
\begin{enumerate}
\item $E$ is 2-coherent.
\item $\hRm$ is a 2-preorder for $E$.
\item There exists a 2-preorder for $E$.
\end{enumerate}
Moreover, the set of 2-preorders for $E$ is closed under non-empty intersections and, if it is non-empty, has $\hRm$ as its smallest member. In this case we have $\leq_X = Z_X$ and $\leq_Y = Z_Y$, where $Z_X$ and $Z_Y$ are as in Definition \ref{D:hRm}.
\end{thm}
\begin{proof}
Suppose first that $E$ is 2-coherent. As we know from Theorem \ref{T:1-pol} that $\hRm$ is a 1-preorder, we need only show that the maps $\iota_X$ and $\iota_Y$ from the diagram in Figure \ref{F:fix1} are order reflecting. But this is immediate from \ref{C5} and \ref{C6}, which amount to stating that $Z_X\subseteq \leq_X$ and $Z_Y\subseteq \leq_Y$, respectively. Thus (1)$\implies$(2). 

That (2)$\implies$(3) is automatic, so suppose now that $\preceq$ is a 2-preorder for $E$. Then, given $x_1,x_2\in X$ and $p\in P$ with $x_1 \R  e_Y(p)$ and $e_X(p)\leq_X x_2$, as $\preceq$ is a 2-preorder (so necessarily a 1-preorder), we have $x_1\preceq e_Y(p)\preceq e_X(p)\preceq x_2$, and thus $x_1\preceq x_2$ by transitivity. It follows immediately from the assumption that $\iota_X$ is an order embedding that $x_1\leq_X x_2$. This proves that \ref{C5} holds for $E$, and that \ref{C6} also holds follows by a dual argument.

That the set of 2-preorders is closed under non-empty intersections is easy to see, and that $\hRm$ is its smallest member (assuming it has any) follows immediately from the corresponding statement about 1-preorders made as part of Theorem \ref{T:1-pol}. Finally, we have noted that \ref{C5} and \ref{C6} imply that $Z_X\subseteq \leq_X$ and $Z_Y\subseteq \leq_Y$, and the opposite inclusions are automatic, so we are done.  
\end{proof}

We will provide examples showing that the strengths of the coherence conditions defined so far are strictly increasing, but we defer this till Section \ref{S:strict}.

\begin{defn}\label{D:3-pre}
Let $E=\Pe$ be an extension polarity. Define a \textbf{3-preorder} for $E$ to be a 2-preorder for $E$ such that the maps $\iota_X$ and $\iota_Y$ from the diagram in Figure \ref{F:fix1} satisfy the following conditions:
\begin{description}
\item[(P4)\label{P4}] For all $S\subseteq P$, if $\bw e_X[S]$ exists in $X$, then $\iota_X(\bw e_X[S]) = \bw \iota_X\circ e_X[S]$.
\item[(P5)\label{P5}] For all $T\subseteq P$, if $\bv e_Y[T]$ exists in $Y$, then $\iota_Y(\bv e_Y[T]) = \bv \iota_Y\circ e_Y[T]$.
\end{description}
\end{defn}

\begin{defn}\label{3-co}
Let $E=\Pe$ be an extension polarity. We say $E$ is \textbf{3-coherent} if it is 2-coherent and also satisfies the following conditions:
\begin{description}
\item[(C7)\label{C7}] \begin{align*}&(\forall x\in X)(\forall y_1,y_2\in Y)(\forall S\subseteq P)\\ &\Big(\big(\bw e_X[S]= x \amp x \R  y_2 \amp \forall p\in S\big(y_1\leq_Y e_Y(p)\big)\big)\ra y_1\leq_Y y_2\Big).\end{align*} 
\item[(C8)\label{C8}] \begin{align*}&(x_1,x_2\in X)(\forall y\in Y)(\forall T\subseteq P)\\ & \Big(\big(\bv e_Y[T] = y \amp x_1 \R  y \amp \forall q\in T \big(e_X(q)\leq_X x_2\big)\big)\ra x_1\leq_X x_2\Big).\end{align*}
\end{description}
\end{defn}

\begin{defn}[$\hRg$]\label{D:hRg}
Let $\Pe$ be an extension polarity. Define the relation $\hRg\subseteq (X\cup Y)^2$ to be the union of $\Rz$ with the sets $Z_S$ and $Z_T$ defined below.
\begin{itemize}
\item[] \begin{align*}Z_S = \{&(y,x)\in Y\times X : \exists S\subseteq P\Big(\bw e_X[S]\text{ exists in $X$}, \bw e_X[S]\leq_X x\\ &\text{ and }\forall p\in S\big(y\leq_Y e_Y(p) \big) \Big)\}.\end{align*} 
\item[] \begin{align*}Z_T = \{&(y,x)\in Y\times X : \exists T\subseteq P\Big(\bv e_Y[T]\text{ exists in $Y$}, y\leq_Y \bv e_Y[T] \\ &\text{ and }\forall q\in T\big(e_X(q)\leq_X x \big) \Big)\}.\end{align*} 
\end{itemize}
\end{defn}

\begin{lemma}\label{L:upgrade}
Let $E=\Pe$ be a 3-coherent extension polarity, and let $\preceq$ be a 2-preorder for $E$. Then $\preceq$ is also a 3-preorder for $E$ if and only if $Z_S\cup Z_T\subseteq \preceq$.
\end{lemma}
\begin{proof}
Suppose first that $\preceq$ is a 3-preorder for $E$, let $S\subseteq P$, and suppose $\bw e_X[S]$ exists. Suppose also that $y\leq_Y e_Y(p)$ for all $p\in S$. Then, by the assumption that $\preceq$ is a 3-preorder (and so necessarily a 1-preorder) it follows that $y\preceq e_X(p)$ for all $p\in S$, and thus that $\iota_Y(y)$ is a lower bound for $\iota_X\circ e_X[S]$. It then follows from the meet-preservation property of 3-preorders that $y\preceq \bw e_X[S]$ as claimed. Thus $Z_S\subseteq \preceq$. That $Z_T\subseteq \preceq$ follows from a dual argument.

For the converse, let $\preceq$ be a 2-preorder for $E$ and suppose first that $Z_S\subseteq \preceq$. Let $S\subseteq P$ and suppose $\bw e_X[S]$ is defined in $X$. Then $\iota_X(\bw e[S])$ is obviously a lower bound for $\iota_X\circ e_X[S]$. Let $z\in \XUpY$ and suppose $z$ is also a lower bound for $\iota_X\circ e_X[S]$. If $z\in \iota_X[X]$, then we must have $z\leq \iota_X(\bw e[S])$, as $\iota_X$ is an order embedding. Moreover, if $z = \iota_Y(y)$ for some $y\in Y$, then we have $y\preceq e_X(p)$ for all $p\in S$, and it follows from the fact that $\preceq$ is a 2-preorder that $y\leq_Y e_Y(p)$ for all $p\in S$. Consequently, that $y\preceq \bw e_X[S]$, and thus that $z \leq \iota_X(\bw e_X[S])$, follows from the definition of $Z_S$. So $\preceq$ satisfies \ref{P4}, and thus also \ref{P5} by duality. 
\end{proof}

\begin{thm}\label{T:3-pol}
Given an extension polarity $E=\Pe$, the following are equivalent:
\begin{enumerate}
\item $E$ is 3-coherent.
\item $\hRg$ is a 3-preorder for $E$.
\item There exists a 3-preorder for $E$.
\end{enumerate}
Moreover, the set of 3-preorders for $E$ is closed under non-empty intersections and, if it is non-empty, has $\hRg$ as its smallest member. 
\end{thm}
\begin{proof}
Suppose first that $E$ is 3-coherent.   It's apparent from the definition of $\hRg$ that it is reflexive and that the $\iota$ maps will be order embeddings. Now, we have $e_X(p)\hRg e_Y(p)$ from \ref{C3} and the definition of $\hRg$, and by setting $S= \{p\}$ and $x = e_X(p)$, we can get $e_Y(p) \hRg e_X(p)$ from the fact that $(e_Y(p),e_X(p))\in Z_S$. Thus the diagram in Figure \ref{F:fix1} will commute. So, if $\hRg$ is transitive, then it is a 2-preorder, and thus a 3-preorder, for $E$ by Lemma \ref{L:upgrade}.

The main work now is showing that $\hRg$ is transitive. Again this breaks down into eight cases of form $(z_1,z_2,z_3)$. The cases where a $y$ value does not appear before an $x$ value are covered by the proof of Theorem \ref{T:0-pol}, so the proofs need not be repeated. There are four remaining cases.

\begin{enumerate}[$\bullet$]
\item $(y,x_1,x_2)$: We have $x_1\leq_X x_2$, and two subcases.
\begin{enumerate}[$-$]
\item Suppose first that $(y,x_1)\in Z_S$. So there is $S\subseteq P$ with $\bw e_X[S] \leq_X x_1$ and $y\leq_Y e_Y(p)$ for all $p\in S$. Then, since $x_1\leq_X x_2$ we also have $\bw e_X[S] \leq_X x_2$, and so $(y,x_2)\in Z_S$ too. 
\item Suppose instead that $(y,x_1)\in Z_T$. Then there is $T\subseteq P$ with $\bv e_Y[T] \geq_Y y$ and $e_X(q)\leq_X x_1$ for all $q\in T$. Then, as $x_1\leq_X x_2$ we also have $e_X(q) \leq_X x_2$ for all $q\in T$, and so $(y,x_2)\in Z_T$ too.
\end{enumerate}

\item $(y_1,y_2,x)$: Dual to the previous case. 

\item $(x_1,y,x_2)$: We have $x_1 \R  y$ and two subcases.
\begin{enumerate}[$-$]
\item Suppose first that $(y,x_2)\in Z_S$. Then there is $S\subseteq P$ with $\bw e_X[S] \leq_X x_2$ and $y\leq_Y e_Y(p)$ for all $p\in S$. So, given $p\in S$, as $x_1 \R y$ by assumption, we have  $x_1 \R  e_Y(p)$ by \ref{C2}. It then follows from \ref{C5} that $x_1\leq_X e_X(p)$, and so $x_1\leq_X \bw e_X[S] \leq_X x_2$ as required.
\item Suppose now that $(y,x_2)\in Z_T$. Then there is $T\subseteq P$ with $y\leq_Y \bv e_Y[T] $ and $e_X(q)\leq_X x_2$ for all $q\in T$. Then we have $x_1\R\bv e_Y[T]$ by \ref{C2}, and so $x_1\leq_X x_2$ by \ref{C8}.
\end{enumerate}

\item Dual to the previous case.
\end{enumerate}

This proves that $\hRg$ is a 3-preorder, and thus (1)$\implies$(2). Again, that (2)$\implies$(3) is automatic, so suppose now that $\preceq$ is a 3-preorder for $E$. If \ref{C7} were to fail for some $y_1,y_2\in Y$, then we would have $y_1\preceq y_2$ (via an appeal to \ref{P4} and other properties of $\preceq$), but not $y_1\leq_Y y_2$, which would contradict the fact that $\preceq$ is a 3-preorder. That we also have \ref{C8} follows from a dual argument. Thus (3)$\implies$(1).

Finally, it is again easy to show that the set of 3-preorders on $E$ will be closed under non-empty intersections, and that $\hRg$ is its smallest element whenever it is non-empty follows immediately from Lemma \ref{L:upgrade}.
\end{proof}

Note that, when $\Pe$ is 2-coherent, given $x\in X$ and $y\in Y$, and given $p,q\in P$ such that
\begin{enumerate}
\item $y\leq_Y e_Y(p)$,
\item $e_X(p)\R  e_Y(q)$, and 
\item $e_X(q)\leq_X x$, 
\end{enumerate}
by setting $S = \{q\}$ we have $\bw e_X[S]\leq_X x$, and also $y\leq_Y e_Y(q)$ by \ref{C5}. Recalling Definitions \ref{D:hRm} and \ref{D:hRg}, it follows that $Z_{YX}\subseteq Z_S$, and we also have $Z_{YX}\subseteq Z_T$ by a dual argument. Example \ref{E:Psep}, later, demonstrates that these inclusions may be strict, as, even when $E=\Pe$ is 3-coherent, there may be a 2-preorder $\preceq$ for $E$ that is not also a 3-preorder for $E$. In that example we have $Z_S\cap Z_T\not\subseteq \preceq$, but we must have $Z_{YX}\subseteq \preceq$ as, by Theorem \ref{T:2-pol}, we have $\hRm\subseteq \preceq$. Thus we cannot have either $Z_S\subseteq Z_{YX}$ or $Z_T\subseteq Z_{YX}$.

\section{Galois polarities}\label{S:Galois}
\subsection{Entanglement}
In applications of polarities to completion theory, the orders on the sets $X$ and $Y$ of an order polarity $\PXY$ are related to $\R $ via a property we present here as Definition \ref{D:ent}. 

\begin{defn}\label{D:ent}
If $\Pe$ is an extension polarity, we say it is \textbf{entangled} if the following conditions are satisfied:
\begin{description}
\item[(E1)\label{E1}] For all $x_1\not\leq x_2\in X$ there is $y\in Y$ with $(x_2,y)\in \R $ and $(x_1,y)\notin \R$.
\item[(E2)\label{E2}] For all $y_1\not\leq y_2\in Y$ there is $x\in X$ with $(x,y_1)\in \R  $ and $(x,y_2)\notin \R  $.
\end{description}
In this situation we also say that $\Pe$ is an \textbf{entangled polarity}.
\end{defn}

For entangled polarities we can refine Definition \ref{D:0-co} using the following lemma.

\begin{lemma}\label{L:ent}
Let $\PXY$ be an entangled order polarity. Then $\PXY$ is 0-coherent if and only if:
\begin{description}
\item[(C1$'$)\label{C1'}]  $\forall x_1,x_2\in X\Big(x_1\leq_X x_2 \lra \forall y\in Y\big(x_2\R  y\ra x_1\R y\big)\Big)$.
\item[(C2$'$)\label{C2'}] $\forall y_1,y_2\in Y\Big(y_1\leq_Y y_2\lra \forall x\in X\big(x\R  y_1\ra x\R y_2\big)\Big)$.
\end{description}
\end{lemma} 
\begin{proof}
We claim that \ref{C1'} and \ref{C2'} here are equivalent, respectively, to \ref{C1} and \ref{C2} when $\Pe$ is entangled. This is immediate from the definitions. 
\end{proof}

In the case of entangled polarities, using \ref{C1'} and \ref{C2'} we could, if we were so inclined, restate things like the various coherence conditions to avoid explicit reference to the orders on $X$ and $Y$. Lemma \ref{L:ent} also has the following corollary.

\begin{cor}\label{C:ent}
Let $E=\Pe$ be an entangled extension polarity. Then, if $\preceq$ is a 0-preorder for $E$, for all $x_1,x_2\in X$ we have $x_1\leq_X x_2 \iff x_1\preceq x_2$, and for all $y_1,y_2\in Y$ we have $y_1\leq_Y y_2\iff y_1\preceq y_2$. Similarly, if $\preceq$ is a 1-preorder for $E$, then it is also a 2-preorder for $E$. 
\end{cor}
\begin{proof}
Let $\preceq$ be a 0-preorder for $E$.  Appealing to Lemma \ref{L:ent} we assume that \ref{C1'} and \ref{C2'} both hold. Let $x_1\not\leq_X x_2\in X$. Then, by entanglement, there is $y\in Y$ with $(x_2,y)\in \R$ and $(x_1,y)\notin \R $. So we cannot have $x_1 \preceq x_2$, as otherwise transitivity would produce $x_1 \preceq y$, and consequently $x_1 \R  y$. This proves the first claim. The second claim also follows from this argument, as the difference between 1-preorders and 2-preorders is only that in the latter case the induced maps $\iota_X$ and $\iota_Y$ from Figure \ref{F:fix1} must be order embeddings.   
\end{proof}

Note that, if we treat order polarities as extension polarities where $P$ is the empty poset, then given an order polarity $O=\PXY$, every 0-preorder for $O$ is automatically a 1-preorder. Thus if $O$ is an entangled order polarity, the above result shows that its sets of 0-, 1-, and 2-preorders coincide. However, it is not the case that every 2-preorder for $O$ is necessarily a 3-preorder, as the following example demonstrates. 

\begin{ex} Consider the order polarity $O=(\{x\},\{y\},\{(x,y)\})$ and the preorder $\Rz$ from Definition \ref{D:Rz}. Then $\Rz$ is trivially a 2-preorder for $O$, but it is not a 3-preorder for $O$ as, for example, 
\[\bw_X e_X[\emptyset] = x\not\succeq_0 y = \bw_{X\uplus_{\Rz} Y} \emptyset = \bw_{X\uplus_{\Rz} Y}\iota_X\circ e_X[\emptyset],\]
and thus $\Rz$ does not satisfy \ref{P4}. 
\end{ex}

\subsection{Defining Galois polarities}
\begin{defn}
A \textbf{Galois polarity} is a 3-coherent extension polarity $\Pe$ such that $e_X:P\to X$ is a meet-extension, and $e_Y:P\to Y$ is a join-extension.
\end{defn}

The motivation for the name \emph{Galois polarity} will become clear in Section \ref{S:GC}. Galois polarities have several strong properties, as we shall see. We will use the following lemma.

\begin{lemma}\label{L:conds}
Let $\Pe$ be an extension polarity. Suppose $\Pe$ satisfies \ref{C3}. Then, if $\Pe$ satisfies \ref{C1}, it also satisfies $(\dagger_1)$ below. Similarly, if $\Pe$ satisfies \ref{C2}, then it also satisfies $(\dagger_2)$.
\begin{itemize}
\item[($\dagger_1$)] $(\forall p\in P)(\forall x\in X) \Big(x \leq_X e_X(p)\ra x\R  e_Y(p)\Big)$. 
\item[($\dagger_2$)] $(\forall p\in P)(\forall y\in Y) \Big(e_Y(p) \leq_Y y \ra e_X(p) \R  y\Big)$.
\end{itemize}
Moreover, if a polarity $\Pe$ satisfies either $(\dagger_1)$ or $(\dagger_2)$, then it also satisfies \ref{C3}.  
\end{lemma}
\begin{proof}
Suppose $\Pe$ satisfies \ref{C1} and \ref{C3}, and let $x\leq_X e_X(p)$ for some $x\in X$ and $p\in P$. Then $e_X(p) \R  e_Y(p)$ by \ref{C3}, and so $x \R  e_Y(p)$ by \ref{C1}. Thus $\Pe$ satisfies ($\dagger_1$). The case where we assume \ref{C2} and \ref{C3} to prove ($\dagger_2$) is dual.  
Suppose now that $\Pe$ satisfies ($\dagger_1$), and let $p\in P$. Then, as $e_X(p)\leq_X e_X(p)$, we have $e_X(p) \R  e_Y(p)$ by ($\dagger_1$), and thus $\Pe$ satisfies \ref{C3}. The case where we assume ($\dagger_2$) and prove \ref{C3} is again dual. 
\end{proof}

\begin{lemma}\label{L:Gent}
Galois polarities are entangled.
\end{lemma}
\begin{proof}
Let $\Pe$ be a Galois polarity, and let $x_1\not\leq x_2\in X$. Then, as $e_X$ is a meet-extension there is $p\in P$ with $x_1\not\leq_X e_X(p)$, and $x_2\leq_X e_X(p)$. Thus $x_2 \R  e_Y(p)$ by $(\dagger_1)$ of Lemma \ref{L:conds}. Moreover, if $x_1 \R  e_Y(p)$, then $x_1 \leq_X e_X(p)$ by \ref{C5}, which contradicts the choice of $p$. We conclude that \ref{E1} holds. A dual argument works for \ref{E2}.
\end{proof}

\begin{cor}\label{C:ent2}
If $G=\Pe$ is a Galois polarity, then every 1-preorder for $G$ is also a 2-preorder for $G$.
\end{cor}
\begin{proof}
This follows immediately from Lemma \ref{L:Gent} and Corollary \ref{C:ent}.
\end{proof}

If $G$ is a Galois polarity, then $\hRg$ from Definition \ref{D:hRg} is the only 3-preorder for $G$ as we show in Theorem \ref{T:unique}. First, the following technical lemma will be useful.

\begin{lemma}\label{L:GalSimp}
If $E=\Pe$ is 3-coherent, then a) and b) below each imply c), for all $x\in X$ and for all $y\in Y$. Moreover, if $E$ is Galois, then a), b) and c) are all equivalent, for all $x$ and $y$.  
\begin{enumerate}[a)]
\item There is $S\subseteq P$ with $\bw e_X[S] \leq_X x$ and $y\leq_Y e_Y(p)$ for all $p\in S$.
\item There is $T\subseteq P$ with $\bv e_Y[T] \geq y$ and $e_X(q)\leq_X x$ for all $q\in T$.  
\item For all $p,q\in P$, if $e_Y(p)\leq_Y y$ and $x\leq_X e_X(q)$, then $p\leq_P q$.
\end{enumerate}
\end{lemma}
\begin{proof}
 As $E$ is 3-coherent, we can let $\preceq$ be a 3-preorder for $E$.  Suppose first that a) holds for $x$ and $y$, and let $p,q\in P$ with $e_Y(p)\leq_Y y$ and $x\leq_X e_X(q)$. Then we have
\[e_Y(p)\preceq y \preceq \bw e_X[S] \preceq x \preceq e_X(q)\]
for some $S\subseteq P$, using the fact that $Z_S\subseteq \preceq$, by Lemma \ref{L:upgrade}. By commutativity of the diagram in Figure \ref{F:fix1} we must therefore have $e_X(p)\preceq e_X(q)$, and so $p\leq_P q$. This shows a)$\implies$ c), and a dual argument shows b)$\implies$ c).

Suppose now that $E$ is Galois and that c) holds for $x$ and $y$. As $E$ is Galois we have $x = \bw e_X[S]$ for $S = e_X^{-1}(x^\uparrow)$, and $y = \bv e_Y[T]$ for $T=e_Y^{-1}(y^\downarrow)$. By c) we have $q\leq_P p$ for all $q\in T$ and $p\in S$. Given a 3-preorder $\preceq$ for $E$ we thus have $e_X(q) \preceq e_Y(p)$ for all $q\in T$ and $p\in S$, and, appealing to \ref{P4} and \ref{P5}, we must have $y = \bv e_Y[T] \preceq \bw e_X[S] = x$, and thus $y\leq_Y e_Y(p)$ for all $p\in S$, and $e_X(q)\leq_X x$ for all $q\in T$. It follows that c) implies both a) and b), and so we have the claimed equivalence.
\end{proof}

\begin{defn}\label{D:alt}
Given an extension polarity $\Pe$, define 
\[Z'_{YX} = \{(y,x)\in Y\times X: (\forall p\in e_Y^{-1}(y^\downarrow)(\forall q\in e_X^{-1}(x^\uparrow))\big(p\leq_P q \big)  \}.\]
\end{defn}

\begin{cor}\label{C:alt}
If $G = \Pe$ is a Galois polarity, then we can define $\hRg$ from Definition \ref{D:hRg} to be $\Rz\cup Z'_{YX}$.
\end{cor}
\begin{proof}
It follows immediately from Lemma \ref{L:GalSimp} that $Z'_{YX}=Z_S\cup Z_T$ in this case.
\end{proof}

\begin{thm}\label{T:unique}
If $G=\Pe$ is a Galois polarity, then:
\begin{enumerate}
\item The maps $\iota_X:X\to \XUgY$ and $\iota_Y:Y\to \XUgY$ are completely meet- and join-preserving respectively.
\item $\hRg$ is the only 3-preorder for $G$.
\item $X\uplus_{\hRg} Y$ is join-generated by $\iota_X[X]$, and meet-generated by $\iota_Y[Y]$. 
\end{enumerate}
\end{thm}
\begin{proof}
We will start by showing that $\iota_X$ is completely meet-preserving. Let $x\in X$ and suppose $x = \bw Z$ for some $Z\subseteq X$. For each $z\in Z$ define $S_z= \bw e_X^{-1}(z^\uparrow)$. Then for all $z\in Z$ we have $z = \bw e_X[S_z]$, as $e_X$ is a meet-extension. Moreover, $x = \bw e_X[\bigcup_{z\in Z} S_z]$. So, using \ref{P4},
\[\iota_X (x) = \iota_X(\bw e_X[\bigcup_{z\in Z} S_z]) = \bw \iota_X\circ e_X[\bigcup_{z\in Z} S_z]= \bw_{z\in Z} \iota_X(\bw e_X[S_z])=\bw \iota_X[Z].\]
This shows $\iota_X$ is completely meet-preserving, and that $\iota_Y$ is completely join-preserving follows from a dual argument.

To see that $\hRg$ is the only 3-preorder for $G$ note first that it must be the smallest such preorder, by Theorem \ref{T:3-pol}. Moreover, if $\preceq$ is another 3-preorder for $G$, then $\preceq$ is determined, by $\leq_X$, $\leq_Y$, and $\R $, everywhere except on $Y\times X$. So $\preceq\neq\hRg$ if and only if there is $x\in X$ and $y\in Y$ with $y\preceq x$ and $y\nothRg x$. But, by Corollary \ref{C:alt}, this is impossible, as for any $p,q\in P$ with $e_Y(p)\leq_Y y$ and $x\leq_X e_X(q)$ we are forced to have $p\leq_P q$ by the transitivity of $\preceq$ and the commutativity of the diagram in Figure \ref{F:fix1}. 

Finally, $\iota_X[X]$ is a join-dense subset of $X\uplus_{\hRg} Y$ because it's a join-dense subset of itself, and it contains $\iota_X\circ e_X[P]$, which is a join-dense subset of $\iota_Y[Y]$.    
\end{proof}

Given a  0-coherent extension polarity $E=\Pe$ where $e_X$ and $e_Y$ are meet- and join-extensions respectively, there is a simple necessary and sufficient condition for $E$ to be Galois, as explained in the next proposition. 
 
\begin{prop}\label{P:GalSimp}
Let $E=\Pe$ be 0-coherent, and let $e_X$ and $e_Y$ be, respectively, meet- and join-extensions of $P$. Then $\Pe$ is Galois if and only if the following conditions are both satisfied:
\begin{description}
\item[(S1)\label{S1}] $(\forall p\in P)(\forall x\in X)\Big(x\leq_X e_X(p)\lra x \R e_Y(p)\Big)$.
\item[(S2)\label{S2}] $(\forall p\in P)(\forall y\in Y)\Big(e_Y(p)\leq_Y y\lra e_X(p) \R y\Big)$.
\end{description}
\end{prop}
\begin{proof}
Suppose first that $E$ is Galois, and let $p\in P$ and $x\in X$. Suppose $x\leq_X e_X(p)$. Then $x \R e_Y(p)$ by Lemma \ref{L:conds}. Conversely, if $x\R e_Y(p)$, then $x\leq_X e_X(p)$ by\ref{C5}, as $e_X(p)\leq_X e_X(p)$. Thus (S1) holds, and (S2) holds by a dual argument.

Suppose now that $E$ is 0-coherent and satisfies (S1) and (S2), and also that $e_X$ and $e_Y$ are meet- and join-extensions respectively. We will show that the conditions \ref{C3}--\ref{C8} are satisfied.
\begin{enumerate}

\item[\ref{C3}:] This follows immediately from (S1) as $e_X(p)\leq _X e_X(p)$ for all $p\in P$.

\item[\ref{C4}:] If $x\R e_Y(p)$ and $e_X(p)\R y$, then $x\leq_X e_X(p)$ by (S1), and so $x\R y$ by \ref{C1}.

\item[\ref{C5}:] Let $x_1\R e_Y(p)$ and let $e_X(p)\leq_X x_2$. Then $x_1\leq_X e_X(p)$ by \ref{S1}, and so $x_1\leq_X x_2$ by transitivity of $\leq_X$. 

\item[\ref{C6}:]  This is dual to (C5).

\item[\ref{C7}:] Let $\bw e_X[S] = x$, let $x\R y_2$, and suppose $y_1\leq_Y e_Y(p)$ for all $p\in S$. Let $q\in P$ and suppose $e_Y(q)\leq_Y y_1$. Then $q\leq_P p$ for all $p\in S$, as $e_Y$ is an order embedding, and so $e_X(q)\leq_X x$. Thus $e_X(q) \R y_2$ by \ref{C1}, and so $e_Y(q)\leq_Y y_2$ by \ref{S2}. As $e_Y$ is a join-extension it follows that $y_1\leq_Y y_2$ as required.

\item[\ref{C8}:] This is dual to (C7).
\end{enumerate}
\end{proof}

It follows from Proposition \ref{P:GalSimp} that what we call a Galois polarity corresponds to what \cite[Section 4]{GJP13} calls a $\Delta_1$-polarity. See also \cite[Proposition 4.1]{GJP13}, which tells us that the preorder $\hRg$ as defined using Corollary \ref{C:alt} is the one arising naturally from $\GXY$. Theorem \ref{T:unique} says that this is in fact the only 3-preorder definable for a Galois polarity. Note that if $\Pe$ is not Galois, then $\hRg$ may not be a preorder.

Since Galois polarities have only one 3-preorder, to lighten the notation we will from now on write e.g. $X\uplus Y$ in place of $\XUgY$ when working with Galois polarities.

\section{The satisfaction and separation of the coherence conditions}\label{S:satisfaction}
\subsection{Sets of coherent relations}\label{S:satisfaction1}
If $X$ and $Y$ are posets, it's easy to see that the set of relations on $X\times Y$ such that the induced order polarity is 0-coherent is closed under arbitrary unions and intersections, and has $\emptyset$ and $X\times Y$ as least and greatest elements respectively. The situation for extension polarities and more restrictive forms of coherence is a little more delicate, as illustrated by Proposition \ref{P:relations} below. First we introduce another definition.

\begin{defn}[$\R_l$]\label{D:Rl}
Let $e_X:P\to X$ and $e_Y:P\to Y$ be poset extensions. Define the relation $\R_l\subseteq X\times Y$ by
\[x \R_l y \iff e_X^{-1}(x^\uparrow)\cap e_Y^{-1}(y^\downarrow)\neq\emptyset.\]
\end{defn}

\begin{prop}\label{P:relations}
Let $e_X:P\to X$ and $e_Y:P\to Y$ be poset extensions. Then:
\begin{enumerate}
\item The set of relations $\R$ such that $\Pe$ is $n$-coherent is closed under non-empty intersections for all $n\in\{0,1,2,3\}$.
\item $(e_X,e_Y,\R_l)$ is 2-coherent.
\item If $\R\subseteq X\times Y$ and $\Pe$ is $1$-coherent, then $\R_l\subseteq \R$.
\item If $e_X$ and $e_Y$ are, respectively, meet- and join-extensions, then $(e_X,e_Y,\R_l)$ is 3-coherent (and thus Galois).
\item If $(e_X,e_Y,\R_l)$ is not $3$-coherent, then there is no $\R$ such that $\Pe$ is $3$-coherent. 
\end{enumerate}
\end{prop}
\begin{proof}
First, that the sets in question are all closed under non-empty intersections can be proved by a routine inspection of the conditions \ref{C1}--\ref{C8} and we omit the details. Checking that $(e_X,e_Y,\R_l)$ is 2-coherent is a similarly straightforward check of conditions \ref{C1}--\ref{C6}. Thus we have dealt with (1) and (2). For (3),  If $\R$ is a relation such that $\Pe$ is 1-coherent, then, in particular, $\R$ must satisfy \ref{C1}, \ref{C2} and \ref{C3}. Given $x\in X$ and $y\in Y$, if there is $p\in e_X^{-1}(x^\uparrow)\cap e_Y^{-1}(y^\downarrow)$, then we have $x\leq_X e_X(p)$ and $e_Y(p)\leq_Y y$ by choice of $p$, and $e_X(p)\R e_Y(p)$ by \ref{C3}. Thus $x\R e_Y(p)$ by \ref{C1}, and so $x\R y$ by \ref{C2}. It follows that $\R_l\subseteq \R$ as claimed.

For (4), suppose that $e_X$ is a meet-extension and $e_Y$ is a join-extension. We will check that $\R_l$ also satisfies \ref{C7}. Let $S\subseteq P$, and let $x = \bw e_X[S]$ in $X$. Let $y_1,y_2\in Y$ and suppose that $y_1 \leq_Y e_Y(p)$ for all $p\in S$, and that $x \R_l y_2$. Let $q\in P$ and suppose $e_Y(q)\leq_Y y_1$. Then $e_Y(q)\leq_Y e_Y(p)$, and thus $q\leq_P p$, for all $p\in S$. It follows that $e_X(q)\leq_X e_X(p)$ for all $p\in S$, and so $e_X(q)\leq_X x$. Also, by definition of $\R_l$, there is $q'\in P$ with $x\leq_X e_X(q')$ and $e_Y(q')\leq_Y y_2$. But then $q\leq_P q'$, and consequently $e_Y(q)\leq_Y y_2$. This is true for all $q\in e_Y^{-1}(y_1^\downarrow)$, and so $y_1 \leq_Y y_2$ as $e_Y$ is a join-extension. $\R_l$ also satisfies \ref{C8} by duality, and so the claim is proved.

Finally, if $(e_X,e_Y,\R_l)$ is not 3-coherent, then, as we have shown it must be 2-coherent, it must fail to satisfy either \ref{C7} or \ref{C8}. In either case, since by (3) any $\R$ making $\Pe$ 3-coherent must contain $\R_l$, inspection of \ref{C7} and \ref{C8} reveals that no such $\R$ can exist, which proves (5). 
\end{proof}
 
Note that when $e_Y$ is not a join-extension $(e_X,e_Y,\R_l)$ may not satisfy \ref{C7}, as Example \ref{E:C7fail} demonstrates. By duality, when $e_X$ is not a meet-extension $(e_X,e_Y,\R_l)$ may not satisfy \ref{C8}.  

\begin{ex}\label{E:C7fail}
Let $P$ be the poset in Figure \ref{F:C7failP}, and let $e_X$ and $e_Y$ be the extensions defined in Figures \ref{F:C7failX} and \ref{F:C7failY} respectively. Here the embedded images of elements of $P$ are represented using $\bullet$, and the extra elements of $X$ and $Y$ using $\circ$. Note that $e_X$ is a meet-extension, but $e_Y$ is not a join-extension. Let $S= \{p,q\}$. Then $x = \bw e_X[S]$, and $x\R_l e_Y(r)$. But we also have $y\leq_Y e_Y(p)$ and $y\leq_Y e_Y(q)$, but $y\not\leq_Y e_Y(r)$. So \ref{C7} does not hold for $\R_l$. 
\end{ex} 

\begin{figure}[hbbp]
  \centering
  \begin{minipage}[b]{0.32\textwidth}
  \[\xymatrix@=1.5em{ \bullet_p & \bullet_q & \bullet_r 
}\] 
\caption{}
\label{F:C7failP}
  \end{minipage}
  \hfill
  \begin{minipage}[b]{0.32\textwidth}
   \[\xymatrix@=1.5em{ \bullet_p\ar@{-}[dr] & \bullet_q\ar@{-}[d] & \bullet_r\ar@{-}[dl] \\
	& \circ_x &
}\] 
\caption{}
\label{F:C7failX}
  \end{minipage}
	\hfill
  \begin{minipage}[b]{0.32\textwidth}
   \[\xymatrix@=1.5em{\bullet_p\ar@{-}[d] & \bullet_q\ar@{-}[dl] & \bullet_r \\
	 \circ_{y} 
}\] 
\caption{}
\label{F:C7failY}
  \end{minipage}
\end{figure}

\subsection{A strict hierarchy for coherence}\label{S:strict}

Example \ref{E:C7fail}, taken with Proposition \ref{P:relations}(2), also demonstrates that it is possible for an extension polarity to be 2-coherent but not 3-coherent (take $(e_X, e_Y, \R_l)$ from this example). Thus 3-coherence is a strictly stronger property than 2-coherence. However, this example only applies when either $e_Y$ fails to be a join-extension, or, by duality, when $e_X$ fails to be a meet-extension. Example \ref{E:C7fail2} below demonstrates that, even when $e_X$ and $e_Y$ \emph{are} meet- and join-extensions respectively, there may be choices of $\R$ for which $\Pe$ is 2-coherent but not 3-coherent.

\begin{ex}\label{E:C7fail2}
Let $e_X$ and $e_Y$ be as in Figures \ref{F:C7failX2} and \ref{F:C7failY2} respectively. Then it's easy to see that $e_X$ and $e_Y$ are meet- and join-extensions respectively. Moreover, if we define $\R = \R_l\cup \{(x, y_2)\}$, then $E=\Pe$ is 2-coherent, as can be observed by noting the preorder on $X\cup Y$ described in Figure \ref{F:2cons}. However, $E$ is not 3-coherent, as we prove now. If $E$ is 3-coherent, then it is Galois, by definition, and the characterization of $\hRg$ from Corollary \ref{C:alt} is valid. Noting that $(y_1,x)\in Z'_{YX}$, it follows that $y_1\hRg y_2$, and thus that $\hRg$ is not a 2-preorder for $E$, as it does not reflect the order on $Y$. But then $E$ is not 3-coherent, by Theorem \ref{T:3-pol}, and to avoid contradiction we must conclude that $\Pe$ is not 3-coherent after all.     
\end{ex}

\begin{figure}[!tbbp]
  \centering
  \begin{minipage}[b]{0.32\textwidth}
   \[\xymatrix@=1.5em{ \bullet\ar@{-}[dr] & \bullet\ar@{-}[d] \\ 
	& \circ_x\\
	\bullet\ar@{-}[ur] & \bullet\ar@{-}[u] & \bullet & \bullet
}\] 
\caption{}
\label{F:C7failX2}
  \end{minipage}
	\hfill
  \begin{minipage}[b]{0.32\textwidth}
   \[\xymatrix@=1.5em{\bullet\ar@{-}[d] & \bullet\ar@{-}[dl] & \circ_{y_2} \\ 
	\circ_{y_1} &  \\
	\bullet\ar@{-}[u] & \bullet\ar@{-}[ul] & \bullet\ar@{-}[uu] & \bullet\ar@{-}[uul]
}\] 
\caption{}
\label{F:C7failY2}
  \end{minipage}
	\hfill
	\begin{minipage}[b]{0.32\textwidth}
	\[\xymatrix@=1.5em{\bullet\ar@{-}[dr] & \bullet\ar@{-}[d] & \circ_{y_2} \\ 
	\circ_{y_1}\ar@{-}[ur]\ar@{-}[u] & \circ_{x}\ar@{-}[ur] &  \\
	\bullet\ar@{-}[ur]\ar@{-}[u] & \bullet\ar@{-}[u]\ar@{-}[ul] & \bullet\ar@{-}[uu] & \bullet\ar@{-}[uul]
}\] 
\caption{}
 \label{F:2cons} 
	\end{minipage}
\end{figure} 

2-coherence is also a strictly stronger condition than 1-coherence, as witnessed by Example \ref{E:1not2} below. 
\begin{ex}\label{E:1not2}
Let $P$ be the two element antichain $\{p,q\}$. Define $X\cong Y\cong P$, and let $e_X$ and $e_Y$ be isomorphisms. Define $\R = \R_l\cup \{(e_X(p), e_Y(q))\}$. Let $E = \Pe$. Then $E$ is 1-coherent, as can be proved via Theorem \ref{T:1-pol}, either by writing down a suitable 1-preorder (the one inducing the two element chain $e_X(p)=e_Y(p)\leq e_X(q)= e_Y(q)$), or by formally proving that $\hRm$ is such a thing.  But $E$ is not  2-coherent, which we can intuit by noticing that there's no way the $\iota$ maps from Figure \ref{F:fix1} are going to be order reflecting, or prove formally via Theorem \ref{T:2-pol} by observing that $\leq_X$ is a strict subset of $Z_X$ here.   
\end{ex}

\subsection{Separating the classes of preorders}
We have seen that the classes of extension polarities defined by the coherence conditions are strictly separated. It is also true that, even for a Galois polarity $G$, the set of 3-preorders for $G$ may be strictly contained in its set of 2-preorders. Moreover, for a 3-coherent polarity $E$, it may be the case that the set of 3-preorders for $E$ is strictly contained in the set of its 2-preorders, which is itself strictly contained in its set of 1-preorders (from Corollary \ref{C:ent2} we know this last statement is not true for Galois polarities). This is demonstrated in Examples \ref{E:Psep} and \ref{E:Psep2} respectively.

\begin{ex}\label{E:Psep}
Let $e_X$ and $e_Y$ be as in Figures \ref{F:enotgX} and \ref{F:enotgY} respectively. Then $G=(e_X,e_Y,\R_l)$ is Galois, by Proposition \ref{P:relations}(4), and the order represented in Figure \ref{F:enotgC} is induced by a 2-preorder for $G$, which we call $\preceq$. However, $\preceq$ is not a 3-preorder for $G$ as \ref{P4} and \ref{P5} fail.  
\end{ex}

\begin{figure}[!tbbp]
  \begin{minipage}[b]{0.32\textwidth}
  \[\xymatrix@=1.5em{\bullet & \bullet \\   
	 & \circ_{x}\ar@{-}[u]\ar@{-}[ul]\\
	 \bullet\ar@{-}[ur] & \bullet\ar@{-}[u]
}\] 
\caption{}
\label{F:enotgX}
  \end{minipage}
  \hfill
  \begin{minipage}[b]{0.32\textwidth}
   \[\xymatrix@=1.5em{
	\bullet\ar@{-}[d] & \bullet\ar@{-}[dl] \\
	 \circ_{y} \\
	\bullet\ar@{-}[u] & \bullet\ar@{-}[ul]
}\] 
\caption{}
\label{F:enotgY}
  \end{minipage}\hfill
	\begin{minipage}[b]{0.32\textwidth}
  \[\xymatrix@=1.5em{\bullet & \bullet \\ 
	 \circ_{y}\ar@{-}[ur]\ar@{-}[u] & \circ_{x}\ar@{-}[u]\ar@{-}[ul]  \\
	\bullet\ar@{-}[ur]\ar@{-}[u] & \bullet\ar@{-}[u]\ar@{-}[ul]
}\] 
\caption{}
\label{F:enotgC}
  \end{minipage}
	\end{figure}
	
\begin{ex}\label{E:Psep2}
Let $P$ and $e_X$ be as in Example \ref{E:Psep}, and let $e_Y$ be defined by the diagram in Figure \ref{F:noteY}. Then $E= (e_X,e_Y,\R_l)$ is 3-coherent, because $\hRg$ induces the order illustrated in Figure \ref{F:is-3-co}. However, $E$ has a 2-preorder that is not a 3-preorder (described in Figure \ref{F:enotgC2}), and a 1-preorder that is not a 2-preorder obtained by additionally setting $y_2\preceq y_1$. 
\end{ex}

\begin{figure}[!tbbp]
\begin{minipage}[b]{0.32\textwidth}
  \[\xymatrix@=1.5em{
	\bullet\ar@{-}[d] & \bullet\ar@{-}[dl] \\
	 \circ_{y} & & \circ_{y_2}\\
	\bullet\ar@{-}[u] & \bullet\ar@{-}[ul] & \circ_{y_1}\ar@{-}[u]
}\] 
\caption{}
\label{F:noteY}
  \end{minipage}
  \hfill
  \begin{minipage}[b]{0.32\textwidth}
    \[\xymatrix@=0.62em{
	\bullet\ar@{-}[dr] & & \bullet\ar@{-}[dl] \\
	 & \circ_{y}\ar@{-}[d]  \\
	& \circ_x && \circ_{y_2} \\
	\bullet\ar@{-}[ur] & & \bullet\ar@{-}[ul] & \circ_{y_1}\ar@{-}[u]
}\] 
\caption{}
\label{F:is-3-co}
  \end{minipage}
	 \hfill
  \begin{minipage}[b]{0.32\textwidth}
   \[\xymatrix@=1.5em{\bullet & \bullet \\ 
	 \circ_{y}\ar@{-}[ur]\ar@{-}[u] & \circ_{x}\ar@{-}[u]\ar@{-}[ul] & \circ_{y_2}\ar@{-}[d]  \\
	\bullet\ar@{-}[ur]\ar@{-}[u] & \bullet\ar@{-}[u]\ar@{-}[ul] & \circ_{y_1}
}\] 
\caption{}
\label{F:enotgC2}
  \end{minipage}
	\end{figure}

\section{Extending and restricting polarity relations}\label{S:ExtRes}
\subsection{Extension}
If $e:P\to Q$ is an order extension, then given another order extension $e':Q\to Q'$, the composition $e'\circ e$ is also an order extension. It is natural to ask whether an extension polarity $E=\Pe$ can be extended to something like $E'=(e'_X\circ e_X, e'_Y\circ e_Y, \R')$, and under what circumstances the level of coherence of $E$ transfers to $E'$. This is of particular interest, for example, if we wish to extend $e_X$ and $e_Y$ to completions, as we shall do in Section \ref{S:GC}. The next theorem provides some answers, but first we need a definition.

\begin{defn}[$\bR$]\label{D:bR}
Let $\Pe$ be an extension polarity, let $i_X:X\to\MX$ and $i_Y:Y\to\MY$ be order extensions with $\MX\cap\MY=\emptyset$. Let $\bR$ be the relation on $\MX\times\MY$ defined by
\[x' \bR y' \iff (\exists x\in X)(\exists y\in Y)\Big(x'\leq_{\MX} i_X(x)\amp i_Y(y)\leq_{\MY} y'\amp x \R  y\Big).\]
\end{defn}

\begin{thm}\label{T:Galois}
Let $E=\Pe$ be an extension polarity, let $i_X:X\to\MX$ and $i_Y:Y\to\MY$ be order extensions with $\MX\cap\MY=\emptyset$. Let $\overline{E} =  (i_X\circ e_X, i_Y\circ e_Y, \bR)$. Then:
\begin{enumerate}
\item $\overline{E}$ is 0-coherent.
\item For all $x\in X$ and for all $y\in Y$ we have 
\[x \R  y\ra i_X(x) \bR i_Y(y),\] 
and the converse is true if and only if $E$ is 0-coherent.
\item If $E$ is $n$-coherent, then $\overline{E}$ is $n$-coherent, for $n \in\{1, 2\}$. 
\item If $E$ is Galois, and if $i_X:X\to \MX$ and $i_Y:Y\to\MY$ are meet- and join-extensions respectively, then $\overline{E}$ is also Galois. 
\item Let $(\MX,\MY,\oS)$ be 0-coherent, and suppose 
\[x\R y\ra i_X(x)\oS i_Y(y).\] Then $\bR\subseteq \oS$.
\item If $\overline{E}$ is not $n$-coherent, then there is no $\oS\subseteq \MX\times \MY$ satisfying 
\[x\R y\ra i_X(x)\oS i_Y(y)\] 
such that $(i_X\circ e_X, i_Y\circ e_Y, \oS)$ is $n$-coherent, for $n\in\{2,3\}$.
\end{enumerate}
\end{thm}
\begin{proof}\mbox{}
\begin{enumerate}
\item We check that $(\MX, \MY, \bR)$ is 0-coherent using Definition \ref{D:0-co}. We need only check \ref{C1} as \ref{C2} is dual. Let $x'_1\leq x_2'\in \MX$, let $y'\in \MY$, and suppose $x_2' \bR y'$. Then there are $x\in X$ and $y\in Y$ with $x'_1\leq_{\MX}x'_2\leq_{\MX} i_X(x)$, with $i_Y(y)\leq_{\MY} y'$, and with $x \R  y$. But then $x_1' \bR y'$, by definition of $\bR$, so \ref{C1} holds.

\item If $x \R  y$, then that $i_X(x) \bR i_Y(y)$ follows directly from the definition. Conversely, suppose $\Pe$ is 0-coherent, let $x_1\in X$, let $y_1\in Y$, and suppose $i_X(x_1) \bR i_Y(y_1)$. Then there is $x_2\in X$ and $y_2\in Y$ with $x_1\leq_X x_2$, with $x_2 \R  y_2$, and with $y_2\leq_Y y_1$. It follows from 0-coherence of $\Pe$ that $x_1\R y_1$ as required. Moreover,  $(i_X\circ e_X, i_Y\circ e_Y, \bR)$ is always 0-coherent by (1), so, if the converse holds $E$ inherits 0-coherence from $\overline{E}$.  

\item Now suppose $E$ is 1-coherent. We check that \ref{C3} and \ref{C4} hold for $\overline{E}$.
\begin{enumerate}
\item[\ref{C3}:] Let $p\in P$. Then $e_X(p) \R  e_Y(p)$ as $E$ is 1-coherent, and it follows easily that $i_X\circ e_X(p) \bR i_Y\circ e_Y(p)$. Thus \ref{C3} holds for $\overline{E}$ as required. 
\item[\ref{C4}:] Let $x'\in \MX$, let $y'\in \MY$, and let $p\in P$. Suppose $x' \bR (i_Y\circ e_Y(p))$ and $(i_X\circ e_X(p)) \bR y'$. Then there are $x_1\in X$ and $y_1\in Y$, with $x'\leq_{\MX} i_X(x_1)$, with $x_1 \R  y_1$, and with $i_Y(y_1)\leq_{\MY} i_Y\circ e_Y(p)$, and also $x_2\in X$ and $y_2\in Y$ with $i_X\circ e_X(p)\leq_{\MX} i_X(x_2)$, with $x_2 \R  y_2$, and with $i_Y(y_2)\leq_{\MY} y'$. As $i_X$ and $i_Y$ are order embeddings we have $y_1\leq_Y e_Y(p)$ and $e_X(p)\leq_X x_2$. As $E$ is 1-coherent it follows from Theorem \ref{T:1-pol} that $\hRm$ is a 1-preorder for $E$, and thus 
\[x_1\hRm y_1 \hRm e_Y(p)\hRm e_X(p)\hRm x_2\hRm y_2. \]
So $x_1 \R  y_2$ by transitivity of $\hRm$ and the fact that it agrees with $\R $ on $X\times Y$. It follows immediately that $x' \bR y'$, and so \ref{C4} holds for $\overline{E}$. 
\end{enumerate}
Thus $\overline{E}$ is 1-coherent. Suppose now that $E$ is 2-coherent. We check that \ref{C5} holds for $\overline{E}$. Let $x_1',x_2'\in \MX$, and let $p\in P$. Suppose $x'_1 \bR (i_Y\circ e_Y(p))$, and $i_X\circ e_X(p)\leq_{\MX} x'_2$. Then there are $x\in X$ and $y\in Y$ with $x'_1\leq_{\MX} i_X(x)$, with $i_Y(y)\leq_{\MY} i_Y\circ e_Y(p)$, and with $x\R  y$. As $E$ is 2-coherent we know from Theorem \ref{T:2-pol} that $\hRm$ is a 2-preorder for $E$, and we have
\[x \hRm y\hRm e_Y(p) \hRm e_X(p).\]
So $x\leq_X e_X(p)$, as $\hRm$ is a 2-preorder,  and consequently 
\[x_1'\leq_{\MX} i_X(x)\leq_{\MX} i_X\circ e_X(p)\leq_{\MX} x'_2.\] 
Thus $x_1'\leq_{\MX} x_2'$, and so \ref{C5} holds. By duality \ref{C6} also holds, and so $\overline{E}$ is 2-coherent as claimed. 

\item Suppose now that $E$ is Galois, and that the $i_X$ and $i_Y$ are meet- and join-extensions respectively. First, that $i_X\circ e_X$ and $i_Y\circ e_Y$ are, respectively, meet- and join-extensions follows from the corresponding properties of $i_X$, $e_X$, $i_Y$ and $e_Y$. It remains only to check that \ref{C7} and \ref{C8} hold for $\overline{E}$. 

Let $x'\in \MX$, let $y_1',y_2'\in \MY$, and let $S\subseteq P$. Suppose $\bw (i_X\circ e_X[S]) = x'$. Suppose also that $x' \bR y_2'$, and that $y_1'\leq_{\MY} i_Y\circ e_Y(p)$ for all $p\in S$. Then there are $x\in X$ and $y\in Y$ with $x'\leq_{\MX} i_X(x)$ and $i_Y(y)\leq_{\MY} y_2'$, and with $x \R  y$. We aim to prove that $y_1'\leq_{\MY} y_2'$. 

Let $y_0\in Y$ be such that $i_Y(y_0)\leq_{\MY} y_1'$, and let $q\in e_Y^{-1}(y_0^\downarrow)$. Then 
\[e_Y(q) \leq_Y y_0 \leq_Y e_Y(p) \text{ for all } p\in S,\] 
and so $i_X\circ e_X(q)\leq_{\MX} x'\leq_{\MX} i_X(x)$, and consequently $e_X(q)\leq_X x$. Since $E$ is Galois, we know from Theorem \ref{T:3-pol} that $\hRg$ is a 3-preorder for $E$, and we have $y_0\hRg x$
as the map $\iota_Y: Y\to X\uplus_{\hRg} Y$ preserves joins of sets in $e_Y[P]$ and $y_0 = \bv e_Y[y_Y^{-1}(y_0^\downarrow)]$. So we have
\[y_0 \hRg x \hRg y,\]  
and thus $y_0\leq_Y y$ for all $y_0$ with $i_Y(y_0)\leq_{\MY} y'_1$. But, as $i_Y$ is a join-extension, we have 
\[y_1' = \bv i_Y[i_Y^{-1}(y_1'^\downarrow)],\] 
and so $y_1'\leq_{\MY} i_Y(y)\leq_{\MY} y'_2$, which is what we are trying to prove. It follows that \ref{C7} holds for $\overline{E}$, and thus by duality \ref{C8} also holds.
 
\item Suppose $x'\bR y'$. Then there is $x\in X$ and $y\in Y$ with $x'\leq_{\MX} i_X(x)$, $x\R y$, and $i_Y(y) \R y'$. Let $\oS\subseteq \MX\times \MY$ satisfy the conditions from (5). Then  $i_X(x) \oS i_Y(y)$, and so $x' \oS y'$ by \ref{C1} and \ref{C2}, and the result follows.
\item From (5) we know that any relation on $\MX\times\MY$ that `extends $\R$' must contain $\bR$. Examination of the conditions \ref{C5}--\ref{C8} reveals that if they fail for $\bR$ they will also fail for any relation containing $\bR$. 

\end{enumerate}  
\end{proof}

Theorem \ref{T:Galois}(6), tells us that if we want to find a 2- or 3-coherent polarity extending $\Pe$, then it suffices to look at $\bR$, as if this does not produce the desired result then nothing will. Note that this does not apply for 1-coherence. To see this, let $P=\{p\}\cong X\cong Y\cong \MX\cong \MY$, and let $\R =\emptyset$. Then \ref{C3} fails for $\overline{E}$, but if $\oS=\{(i_X\circ e_X(p)), i_Y\circ e_Y(p)\}$, then $(i_X\circ e_X,i_Y\circ e_Y,\oS)$ is obviously 1-coherent. 

For 0-coherent polarities we can add converses to some of the statements in Theorem \ref{T:Galois}, but we will leave this till Corollary \ref{C:converses}. Note that for $\overline{E}$ to be 3-coherent it is not sufficient for $E$ to be 3-coherent, or even Galois. The additional restrictions on the extensions $i_X$ and $i_Y$ from Theorem \ref{T:Galois}(4) are necessary, as Example \ref{E:notSuf} demonstrates below.

\begin{ex}\label{E:notSuf}
Let $P$ be the three element antichain from Figure \ref{F:C7failP}, and let $X \cong Y \cong P$. Let $i_X:X\to\MX$ and $i_Y:Y\to\MY$ be the poset extensions illustrated in Figures \ref{F:C7failX} and \ref{F:C7failY} respectively. Define $\R$ on $X\times Y$ by $x\R y\iff$ there is $p\in P$ with $x=e_X(p)$ and $y= e_Y(p)$.
We can put a poset structure on $X\cup Y$ just by identifying copies of elements of $P$ appropriately, in which case we end up with something isomorphic to $P$. Clearly the natural maps $\iota_X$ and $\iota_Y$ are meet- and join-preserving order embeddings here, and so $\Pe$ is Galois.  However, $(i_X\circ e_X, i_Y\circ e_Y, \bR)$ is not 3-coherent. Indeed, it follows from Example \ref{E:C7fail} that, if we define $\oS_l\subseteq \MX\times \MY$ analogously to Definition \ref{D:Rl}, the polarity $(i_X\circ e_X, i_Y\circ e_Y, \oS_l)$ is not 3-coherent. Thus there is no relation $\oS$ such that $(i_X\circ e_X, i_Y\circ e_Y, \oS)$ is 3-coherent, by Proposition \ref{P:relations}(5).
\end{ex}

The following lemma says, roughly, that the extension of the `minimal' polarity relation $\R_l$ is again the minimal polarity relation.
\begin{lemma}\label{L:minExt}
Let $(e_X,e_Y,\R_l)$ be an extension polarity, where $\R_l$ is as in Definition \ref{D:Rl}, and  let $i_X:X\to\MX$ and $i_Y:Y\to\MY$ be order extensions with $\MX\cap\MY=\emptyset$. Then $\overline{\R_l} = \oS_l$, where $\oS_l\subseteq \MX\times \MY$ is defined analogously to $\R_l$.  
\end{lemma}
\begin{proof}
Let $x'\in \MX$ and let $y'\in \MX$. Then
\begin{align*}
x'\overline{\R_l} y' &\iff x'\leq_{\MX} i_X(x)\text{, }x\R_l y\text{ and }i_Y(y)\leq_{\MY} y'\text{ for some } x\in X \text{ and }y\in Y \\
&\iff x'\leq_{\MX} i_X(x)\text{, }i_Y(y)\leq_{\MY} y'\text{ and } e_X^{-1}(x^\uparrow)\cap e_Y^{-1}(y^\downarrow)\neq\emptyset \\
&\iff (i_X\circ e_X)^{-1}(x'^\uparrow)\cap (i_Y\circ e_Y)^{-1}(y'^\downarrow)\neq\emptyset \\
&\iff x' \oS_l y'. 
\end{align*}
\end{proof}

\subsection{Restriction}
If $i_X:X\to \MX$ and $i_Y:Y\to\MY$ are order extensions, then a polarity $(\MX, \MY, \oS)$ can be restricted in a natural way to a polarity $(X, Y, \uS)$. The following definition makes this precise.

\begin{defn}\label{D:uS}
Let $X$ and $Y$ be posets, and let $i_X:X\to \MX$ and $i_Y:Y\to\MY$ be order extensions with $\MX\cap\MY=\emptyset$. Let $\oS$ be a relation on $\MX\times \MY$. Define the relation $\uS\subseteq X\times Y$ by
\[x \uS y \iff i_X(x) \oS i_Y(y).\] 
\end{defn}

It turns out the coherence properties behave quite well under restriction. We make this precise in Theorem \ref{T:phi}, but first we need another definition.

\begin{defn}[$\phi$, $\preceq_\phi$]\label{D:phi}
Let $P$ be a poset, let $e_X:P\to X$, $i_X:X\to \MX$, $e_Y:P\to Y$ and $i_Y:Y\to \MY$ be order extensions, where $X\cap Y = \emptyset = \MX\cap \MY$. Define $\phi':X\cup Y\to \MX\cup\MY$ by
\[\phi'(z) = \begin{cases} i_X(z) \text{ if } z\in X.\\  
 i_Y(y) \text{ if } z\in Y.\end{cases}\] 
Let $\preceq$ be a preorder on $\MX\cup\MY$, and define the preorder $\preceq_\phi$ on $X\cup Y$ by setting \[z_1\preceq_\phi z_2 \iff \phi'(z_1)\preceq \phi'(z_2).\] 

Let $X\uplus_{\preceq_\phi} Y$ and $\MX\uplus_\preceq \MY$ be the posets induced by $X\cup_{\preceq_\phi} Y$ and $\MX\cup_\preceq\MY$ respectively, and let $\iota_X:X\to X\uplus_{\preceq_\phi} Y$, $\iota_Y:Y\to X\uplus_{\preceq_\phi} Y$,  $\iota_{\MX}:\MX\to \MX\uplus_\preceq \MY$ and $\iota_{\MY}:\MY\to \MX\uplus_\preceq \MY$ be the maps induced by the inclusion functions. Define $\phi: X\uplus_{\preceq_\phi} Y\to \MX\uplus_\preceq \MY$ by
\[\phi(z) = \begin{cases} \iota_{\MX}\circ i_X(x) \text{ if $z=\iota_X(x)$ for some $x\in X$.} \\
\iota_{\MY}\circ i_Y(y) \text{ if $z=\iota_Y(y)$ for some $y\in Y$.}
\end{cases}\] 
\end{defn}

It should be reasonably clear that $\phi'$ is well defined. It may not be immediately obvious that $\preceq_\phi$ is a preorder, but a quick check reveals that this is indeed the case. We show that $\phi$ is well defined as part of the next theorem.

\begin{thm}\label{T:phi}
With a setup as in Definition \ref{D:phi}, the map $\phi$ is a well defined order embedding. Moreover, suppose  $\oS\subseteq \MX\times\MY$ and define $E=(i_X\circ e_X,i_Y\circ e_Y, \oS)$. Then: 
\begin{enumerate}
\item If $\preceq$ is a 0-preorder for $E$, then the maps $\iota_X$ and $\iota_Y$ are order preserving.
\item If $\preceq$ is a 1-preorder for $E$, then the diagram in Figure \ref{F:restrict} commutes.
\item If $\preceq$ is a 2-preorder for $E$, the maps $\iota_X$ and $\iota_Y$ are order embeddings.
\item Suppose $\preceq$ is a 3-preorder for $E$, and suppose also that $i_X$ preserves meets in $X$ of subsets of $e_X[P]$ whenever they exist, and that $i_Y$ likewise preserves joins in $Y$ of subsets of $e_Y[P]$. Then $\preceq_\phi$ satisfies \ref{P4} and \ref{P5}.
\end{enumerate}
\end{thm}
\begin{proof}
To see that $\phi$ is well defined and order preserving, suppose $\iota_X(x)\leq \iota_Y(y)$ for some $x\in X$ and $y\in Y$. Then $x \preceq_\phi y$, and thus, by definition of $\preceq_\phi$, we have $i_X(x)\preceq i_Y(y)$. It follows that $\iota_{\MX}\circ i_X(x)\leq \iota_{\MY}\circ i_Y(y)$, and thus $\phi(x)\leq \phi(y)$. By a similar argument, if $\iota_Y(y)\leq \iota_X(x)$, then we also have $\phi(y)\leq \phi(x)$, and so $\phi$ is well defined and order preserving as claimed.

To see that $\phi$ is an order embedding, let $z_1,z_2\in X\uplus_{\preceq_\phi} Y$, and suppose that $\phi(z_1)\leq \phi(z_2)$. There are four cases. Suppose first that $z_1 = \iota_X(x)$ and $z_2 = \iota_Y(y)$ for some $x\in X$ and $y\in Y$. Then $i_X(x)\preceq i_Y(y)$, and so $x\preceq_\phi y$, from which it follows immediately that $\iota_X(x)\leq \iota_Y(y)$, and thus that $z_1\leq z_2$. The other cases are more or less exactly the same.

Suppose $\preceq$ is a 0-preorder for $E$. Then that $\iota_X$ is order preserving follows immediately from the fact that $i_X$ and $\iota_{\MX}$ are order preserving, and $\iota_Y$ is order preserving by duality.

Suppose now that $\preceq$ is a 1-preorder for $E$, and let $p\in P$. Then $i_X\circ e_X(p)\preceq i_Y\circ e_Y(p)$. It follows immediately from this that $e_X(p)\preceq_\phi e_Y(p)$, and thus that $\iota_X\circ e_X(p)\leq \iota_Y\circ e_Y(p)$. By a similar argument we also have $\iota_Y\circ e_Y(p)\leq \iota_X\circ e_X(p)$. This shows the upper left square of the diagram in Figure \ref{F:restrict} commutes, and that the rest of the diagram commutes follows immediately from the definition of $\phi$ and the assumption that $\preceq$ is a 1-preorder for $E$. 

Suppose now that $\preceq$ is a 2-preorder for $E$, let $x_1,x_2\in X$, and suppose that $\iota_X(x_1)\leq \iota_X(x_2)$. Then $\phi(x_1)\leq \phi(x_2)$, and so $\iota_{\MX}\circ i_X(x_1)\leq \iota_{\MX}\circ i_X(x_2)$, by definition of $\phi$. As $\preceq$ is a 2-preorder for $E$, the map $\iota_{\MX}$ is an order embedding, so, since $i_X$ is also an order embedding, we must have $x_1\leq_X x_2$. This shows $\iota_X$ is an order embedding, as we have already proved it is order preserving. The argument for $\iota_Y$ is dual.

Finally, suppose $\preceq$ is a 3-preorder for $E$, and that $i_X$ and $i_Y$ have the preservation properties described above. Let $S\subseteq P$, and suppose $\bw e_X[S]$ exists in $X$. Then $\bw i_X\circ e_X[S] = i_X(\bw e_X[S])$ in $\MX$. Let $z\in X\uplus_{\preceq_\phi} Y$ and suppose $z$ is a lower bound for $\iota_X\circ e_X[S]$. Then $\phi(z)$ is a lower bound for $\phi\circ\iota_X\circ e_X[S]$, and so by commutativity of the diagram in Figure \ref{F:restrict} it follows that $\phi(z)$ is a lower bound for $\iota_{\MX}\circ i_X\circ e_X[S]$. So, as $\preceq$ is a 3-preorder for $E$, we have
\begin{align*}
\phi(z) &\leq \bw \iota_{\MX}\circ i_X\circ e_X[S] \\
&= \iota_{\MX}\circ i_X(\bw e_X[S])\\
& = \phi\circ\iota_X(\bw e_X[S]),
\end{align*} 
and so $z\leq \iota_X(\bw e_X[S])$, as $\phi$ is an order embedding. It follows that $\bw \iota_X\circ e_X[S] = \iota_X(\bw e_X[S])$, and thus $\preceq_\phi$ satisfies \ref{P4}. The argument for \ref{P5} is dual.
\end{proof}

\begin{figure}[htbp]
\[\xymatrix{ P\ar[r]^{e_Y}\ar[d]_{e_X} & Y\ar[d]^{\iota_Y}\ar[r]^{i_Y} & \MY\ar[dd]^{\iota_{\MY}} \\
X\ar[r]_{\iota_X}\ar[d]_{i_X} & X\uplus_{\preceq_\phi} Y\ar[dr]^{\phi} \\
\MX\ar[rr]_{\iota_{\MX}} & & \MX\uplus_\preceq\MY
}\] 
\caption{}
\label{F:restrict}
\end{figure}

\begin{cor}\label{C:rest}
Let $X$ and $Y$ be disjoint posets, and let $i_X:X\to \MX$ and $i_Y:Y\to\MY$ be order extensions with $\MX\cap \MY = \emptyset$. Let $\oS$ be a relation on $\MX\times \MY$. Then:
\begin{enumerate}
\item If $(\MX, \MY, \oS)$ is 0-coherent, then so is $(X,Y,\uS)$. 
\end{enumerate}
Moreover, if $P$ is a poset, and if $e_X:P\to X$ and $e_Y:P\to Y$ are order extensions, then both $E=(i_X\circ e_X, i_Y\circ e_Y, \oS)$ and $\underline{E}=(e_X,e_Y,\uS)$   are extension polarities, and:
\begin{enumerate}
\item[(2)] If $E$ is $n$-coherent, then so is $\underline{E}$ for $n \in\{1,2\}$.
\item[(3)] Suppose $i_X$ preserves meets in $X$ of subsets of $e_X[P]$ whenever they exist, and let $i_Y$ likewise preserve joins in $Y$ of subsets of $e_Y[P]$. Then, if $E$ is $3$-coherent, so is $\underline{E}$, and the same is true if we replace `3-coherent' with `Galois'.
\end{enumerate}
\end{cor}
\begin{proof}
This all almost follows immediately from Theorems \ref{T:0-pol}, \ref{T:1-pol}, \ref{T:2-pol}, \ref{T:3-pol} and \ref{T:phi}, as we have \emph{almost} proved that if $\preceq$ is $n$-preorder for $E$, then $\preceq_\phi$ is an $n$-preorder for $\underline{E}$ for all $n\in\{0,1,2,3\}$ (modulo some extra conditions for $n=3$). To complete the proof we need only show that, for all $x\in X$ and $y\in Y$, we have $x\preceq_\phi y\iff x \uS y$. Now, 
\[x\preceq_\phi y  \iff i_X(x)\preceq i_Y(y) \iff i_X(x) \oS i_Y(y) \iff x\uS y,\]
so we are done. 
\end{proof}

Converses for the implications in Corollary \ref{C:rest} do not hold, as Example \ref{E:notConv} demonstrates.

\begin{ex}\label{E:notConv}
Let $P$ be the poset represented by the $\bullet$ elements in Figure \ref{F:notConv}, let $P\cong X\cong Y\cong\MX$, and let $\MY$ be represented by Figure \ref{F:notConv}. Then the implicit maps $i_X$ and $i_Y$ are obviously meet- and join-extensions respectively, and are also, respectively, trivially completely meet- and join-preserving. Let $\oS = \oS_l\cup \{(p, y)\}$, where $\oS_l\subseteq \MX\times\MY$ is defined analogously to Definition \ref{D:Rl}. Then $\uS= \R_l$,  and so $(e_X, e_Y, \uS)$ is Galois by Proposition \ref{P:relations}(4). However, $(i_X\circ e_X, i_Y\circ e_Y, \oS)$ is not even 0-coherent, as we have $(p, y) \in S$ but $(q, y)\notin S$, and so \ref{C1} fails. 
\end{ex}

\begin{figure}[htbp]
\[\xymatrix@=1.5em{ \circ_y \\
\bullet\ar@{-}[u] & \bullet\ar@{-}[ul] & \bullet_p\\
& & \bullet_q\ar@{-}[u]
}\]
\caption{}
\label{F:notConv}
\end{figure}

Using the notation of Definitions \ref{D:bR} and \ref{D:uS}, we can define a map $\overline{(-)}$ from the complete lattice of relations on $X\times Y$ to the complete lattice of relations on $\MX\times\MY$, by taking $\R$ to $\bR$. Similarly, we can define a map $\underline{(-)}$ going back the other way by taking $\oS$ to $\uS$. These maps are obviously order preserving. We also have the following result.

\begin{lemma}\label{L:Gcon}
 Let $X$ and $Y$ be disjoint posets, let $i_X:X\to\MX$  and $i_Y:Y\to\MY$ be order extensions with $\MX\cap\MY=\emptyset$. Then: 
\begin{enumerate}
\item Let $\R\subseteq X\times Y$. Then $\R\subseteq \underline{(\bR)}$. Moreover, if $(X, Y, \R)$ is 0-coherent, then $\R = \underline{(\bR)}$.
\item Let $\oS\subseteq\MX\times \MY$. If $(\MX, \MY, \oS)$ is 0-coherent, then $\overline{(\underline{\oS})}\subseteq \oS$.  
\end{enumerate} 
\end{lemma}
\begin{proof}
We start with (1). Let $x\in X$, let $y\in Y$ and suppose $x \R y$. Then $i_X(x) \bR i_Y(y)$ by definition of $\bR$, and so $x \underline{(\bR)} y$ by definition of $\underline{(\bR)}$. Suppose now that $(X, Y, \R)$ is 0-coherent and let $x \underline{(\bR)} y$. Then $i_X(x) \bR i_Y(y)$ by definition of $\underline{(\bR)}$, and thus $x \R y$ by Theorem \ref{T:Galois}(2).

For (2), suppose first that $(\MX, \MY, \oS)$ is 0-coherent, and let $x'\in \MX$ and $y'\in \MY$ with $x' \overline{(\underline{\oS})} y'$. Then, by definition of $\overline{(\underline{\oS})}$ there are $x\in X$ and $y\in Y$ with $x'\leq_{\MX} i_X(x)$, with $i_Y(y)\leq_{\MY} y'$, and with $x \underline{\oS} y$. But then $i_X(x) \oS i_Y(y)$ by definition of $\underline{\oS}$, and so $x'\leq_{\MX} i_X(x) \oS i_Y(y)\leq_{\MY} y'$, and thus $x' \oS y'$ by 0-coherence of $(\MX, \MY, \oS)$.
\end{proof}

Note that the opposite inclusion to that in Lemma \ref{L:Gcon}(2) may fail, even when $(i_X\circ e_X, i_Y\circ e_Y, \oS)$ is Galois, as is demonstrated in Example \ref{E:notEq} below. Note also that the polarity $(\MX, \MY, \oS)$ from Example \ref{E:notConv} is not 0-coherent, but, appealing to Lemma \ref{L:minExt}, we have $\overline{(\underline{\oS})}\subseteq \oS$. Thus $(\MX, \MY, \oS)$ being 0-coherent is strictly stronger than having $\overline{(\underline{\oS})}\subseteq \oS$.   

\begin{cor}
Using the notation of Lemma \ref{L:Gcon}, let $L$ be the complete lattice of relations between $X$ and $Y$, and let $M$ be the complete lattice of relations $S\subseteq \MX\times\MY$ such that $(\MX,\MY,S)$ is 0-coherent. Then the maps $\overline{(-)}:L\to M$ and $\underline{(-)}:M\to L$ are, respectively, the left and right adjoints of a Galois connection. 
\end{cor}
\begin{proof}
First, recall the discussion at the start of Section \ref{S:satisfaction1} for the lattice structure of $M$. Moreover,  $\overline{(-)}:L\to M$ is well defined by Theorem \ref{T:Galois}(1). By Lemma \ref{L:Gcon} we have $\R\subseteq \underline{(\bR)}$ for all $R\in L$, and $\overline{(\underline{\oS})}\subseteq \oS$ for all $S\in M$, which is one of the equivalent conditions for two order preserving maps to form a Galois connection (see e.g. \cite[Lemma 7.26]{DavPri02}). 
\end{proof}

Using Corollary \ref{C:rest} and Lemma \ref{L:Gcon} we get partial converses for Theorem \ref{T:Galois}.

\begin{cor}\label{C:converses}
With notation as in Theorem \ref{T:Galois}, let $E= \Pe$ and suppose $E$ is 0-coherent. Let $\overline{E} = (i_X\circ e_X, i_Y\circ e_Y, \bR)$.  Then:
\begin{enumerate}
\item If $\overline{E}$ is $n$-coherent, then so is $E$ for $n \in\{1,2\}$. 
\item  Suppose $i_X$ preserves meets in $X$ of subsets of $e_X[P]$ whenever they exist, and let $i_Y$ likewise preserve joins in $Y$ of subsets of $e_Y[P]$. Then, whenever $\overline{E}$ is $3$-coherent, so is $E$, and this is also true if we replace `3-coherent' with `Galois'. 
\end{enumerate}
\end{cor}
\begin{proof}
For (1), given $n\in\{1,2\}$, if $\overline{E}$ is $n$-coherent, then so is $(e_X,e_Y,\underline{(\bR)})$, by Corollary \ref{C:rest}, and as $E$ is assumed to be 0-coherent we have $\R = \underline{(\bR)}$, by Lemma \ref{L:Gcon}. The proof of (2) is essentially the same. 
\end{proof}

\begin{ex}\label{E:notEq}
Let $P$ be the poset in Figure \ref{F:notEqP}, and let $X \cong Y \cong P$. Let $\MX$ and $\MY$ be the posets in Figures \ref{F:notEqX} and \ref{F:notEqY} respectively. Define $\oS\subseteq\MX\times \MY$ so that $i_X\circ e_X(p) \oS e_Y\circ i_Y(p)$ for all $p\in P$, and also $x' \oS y'$. Then $(i_X\circ e_X, i_Y\circ e_Y, \oS)$ is Galois, as can be seen by considering the poset in Figure \ref{F:notEqPos}, and defining $\iota_{\MX}$ and $\iota_{\MY}$ in the obvious way. However, there is no $x\in X$ and $y\in Y$ with $x'\leq_{\MX} i_X(x)$, with $i_Y(y)\leq_{\MY} y'$, and with $x \underline{\oS} y$. Thus $(x',y')\notin \overline{(\underline{\oS})}$.
\end{ex}

\begin{figure}[!tbbp]
  \centering
  \begin{minipage}[b]{0.32\textwidth}
  \[\xymatrix@=1.5em{ \bullet\ar@{-}[d] & \bullet\ar@{-}[d] \\
 \bullet & \bullet
}\] 
\caption{}
\label{F:notEqP}
  \end{minipage}
  \hfill
  \begin{minipage}[b]{0.32\textwidth}
   \[\xymatrix@=1.5em{ \bullet\ar@{-}[dr]\ar@{-}[d] &  & \bullet\ar@{-}[dl]\ar@{-}[d] \\
\bullet & \circ_{x'} & \bullet
}\] 
\caption{}
\label{F:notEqX}
  \end{minipage}
	 \hfill
  \begin{minipage}[b]{0.32\textwidth}
   \[\xymatrix@=1.5em{ \bullet\ar@{-}[d] & \circ_{y'}  & \bullet\ar@{-}[d] \\
\bullet\ar@{-}[ur] &  & \bullet\ar@{-}[ul]
}\]
\caption{}
\label{F:notEqY}
  \end{minipage}
\end{figure}

\begin{figure}[htbp]
\[\xymatrix@=1.5em{ \bullet\ar@{-}[d]\ar@{-}[dr] & \circ_{y'}  & \bullet\ar@{-}[d]\ar@{-}[dl] \\
\bullet\ar@{-}[ur] & \circ_{x'}\ar@{-}[u]  & \bullet\ar@{-}[ul]
}\]
\caption{}
\label{F:notEqPos}
\end{figure}

\section{Galois polarities revisited}\label{S:GaloiS1}
\subsection{Galois polarities via Galois connections}\label{S:GC}
Galois polarities are so named because their associated (unique) 3-preorder can be described in terms of a Galois connection. This idea is precisely articulated in Corollary \ref{C:Galois} below. 

When $G=\Pe$ is Galois, we know from Theorem \ref{T:unique} that $G$ has only a single 3-preorder, $\hRg$. As mentioned previously, to lighten the notation we write e.g. $X\uplus Y$ in place of $\XUgY$ when working with Galois polarities. The next proposition collects together some useful facts, but first we need a definition. 

\begin{defn}[$\gamma$]\label{D:gamma}
Let $\Pe$ be a Galois polarity. Define $\gamma:P\to X \uplus Y$ by
\[\gamma = \iota_X\circ e_X = \iota_Y\circ e_Y.\]
\end{defn}

Noting Figure \ref{F:fix1}, it's easy to see that $\gamma$ is well defined.

\begin{prop}\label{P:commute2}
Let $G=\Pe$ be a Galois polarity, let $i_X:X\to \MX$ be a completely meet-preserving meet-extension, and let $i_Y:Y\to\MY$ be a completely join-preserving join-extension. Then:
\begin{enumerate}
\item $\overline{G}=(i_X\circ e_X, i_Y\circ e_Y, \bR)$ is Galois.
\item The map $\gamma$ is an order embedding. Moreover, if $S,T\subseteq P$ and $\bw S$ and $\bv T$ exist in $P$, then
\begin{enumerate}
\item $\gamma(\bw S) = \bw \gamma[S] \iff e_X(\bw S)= \bw e_X[S]$, and
\item $\gamma(\bv T) = \bv \gamma[T] \iff e_Y(\bv T)=\bv e_Y[T]$.
\end{enumerate} 
\item $\gamma[P] = \iota_X[X]\cap \iota_Y[Y]$.
\end{enumerate}
\end{prop}
\begin{proof}\mbox{}
\begin{enumerate}[(1)]
\item That $\overline{G}$ is Galois is Theorem \ref{T:Galois}(4).
\item That $\gamma$ is well defined follows from 1-coherence of $\Pe$, and that $\gamma$ is an order embedding follows from 2-coherence of $\Pe$, as $\gamma$ is the composition of two order embeddings, $\iota_X\circ e_X$. That (a) and (b) hold follows from 3-coherence of $\Pe$, as, for example, $\gamma = \iota_X\circ e_X$ and $\iota_X$ preserves meets in $X$ of subsets of $e_X[P]$.
\item We obviously have $\gamma[P] \subseteq \iota_X[X]\cap \iota_Y[Y]$, so let $z\in \iota_X[X]\cap \iota_Y[Y]$. Then there are $x\in X$ and $y\in Y$ with $z=\iota_X(x)=\iota_Y(y)$. Thus, as $\iota_X(x)\leq\iota_Y(y)$ we have $e_X^{-1}(x^\uparrow)\cap e_Y^{-1}(y^\downarrow)\neq\emptyset$. Suppose $p\in e_X^{-1}(x^\uparrow)\cap e_Y^{-1}(y^\downarrow)$, and that $e_X(p)\not\leq_X x$. Then, as $e_X$ is a meet-extension, there is $q\in P$ with $x\leq_X e_X(q)$ and $e_X(p)\not\leq_X e_X(q)$. But this is a contradiction, as, since $\iota_Y(y)\leq \iota_X(x)$, appealing to Corollary \ref{C:alt} we see that that $p\leq_P q$. Thus $x=e_X(p)$, and so $z = \gamma(p)$. It follows that $\iota_X[X]\cap \iota_Y[Y]\subseteq \gamma[P]$ as claimed.  

\end{enumerate}
\end{proof}

Note that we could slightly relax the preservation properties of $i_X$ and $i_Y$ in the above proposition (and also later) to be the same as those in Corollary \ref{C:converses}(2), but this isn't necessary for what we want to do with it. The following fact will be useful.

\begin{prop}\label{P:GaloisExtension}
Let $P$ and $Q$ be posets, let $e_1: P\to J$ be a join-completion, and let $e_2:Q\to M$ be a meet-completion. Then any Galois connection $\alpha:P\leftrightarrow Q:\beta$ extends uniquely to a Galois connection $\alpha': J\leftrightarrow M:\beta'$.
\end{prop}
\begin{proof}
This is \cite[Corollary 2]{Schm74}.
\end{proof}

\begin{defn}[$\mathsf{F},\mathsf{G}$]\label{D:GD}
Let $P$ be a poset, and let $e_X:P\to X$ and $e_Y:P\to Y$ be meet- and join-completions respectively. Define maps $\mathsf{F}:Y\to X$ and $\mathsf{G}:X\to Y$ as follows:
\[\mathsf{F}(y) = \bv e_X[e_Y^{-1}(y^\downarrow)].\]
\[\mathsf{G}(x) = \bw e_Y[e_X^{-1}(x^\uparrow)].\]
\end{defn}

$\mathsf{F}$ and $\mathsf{G}$ are well defined as $X$ and $Y$ are complete.

\begin{lemma}\label{L:Galois}
Let $P$ be a poset, and let $e_X:P\to X$ and $e_Y:P\to Y$ be meet- and join-completions respectively. Then there is a unique Galois connection 
\[\mathsf{F}:Y\leftrightarrow X:\mathsf{G}\] such that $e_X = \mathsf{F}\circ e_Y$ and $e_Y = \mathsf{G}\circ e_X$. Here $\mathsf{F}$ and $\mathsf{G}$ are as in Definition \ref{D:GD}. 
\end{lemma}
\begin{proof}
 Using the fact that $e_X$ and $e_Y$ are, respectively, meet- and join-completions, we have, for all $x\in X$ and for all $y\in Y$,
\begin{align*}
\mathsf{F}(y) \leq_X x &\iff \bv e_X[e_Y^{-1}(y^\downarrow)] \leq_X x \\
&\iff (\forall p\in e_Y^{-1}(y^\downarrow))(\forall q\in e_X^{-1}(x^\uparrow))\Big( p\leq_P q\Big)\\
&\iff y \leq_Y \bw e_Y[e_X^{-1}(x^\uparrow)] \\
&\iff y\leq_Y \mathsf{G}(x). 
\end{align*}

To see that this is the only such Galois connection between $X$ and $Y$ we apply Proposition \ref{P:GaloisExtension} with $P=Q$ and the Galois connection produced by the identity function on $P$.
\end{proof}

\begin{defn}\label{D:Z''}
Let $E=\Pe$ be an extension polarity, let $i_X: X\to \MX$ be a completely meet-preserving meet-completion of $X$, and let $i_Y:Y\to \MY$ be a completely join-preserving join-completion of $Y$. Let $\mathsf{F}:\MY\to\MX$ and $\mathsf{G}:\MX\to\MY$ be as in Definition \ref{D:GD}, with respect to the maps $i_X\circ e_X$, and $i_Y\circ e_Y$. Define
\begin{align*} Z''_{YX}= &\{(y,x)\in Y\times X: \mathsf{F}(i_Y(y))\leq_{\MX} i_X(x)\} \\
= & \{(y,x)\in Y\times X: i_Y(y)\leq_{\MY} \mathsf{G}(i_X(x))\}.\end{align*}
\end{defn}
In the above definition, the maps $\mathsf{F}$ and $\mathsf{G}$ exist as $i_X\circ e_X:P\to \MX$ and $i_Y\circ e_Y:P\to \MY$ are meet- and join-completions respectively.

\begin{cor}\label{C:Galois}
Recall $Z'_{YX}$ from Definition \ref{D:alt}. With a setup as in Definition \ref{D:Z''} we have $Z'_{YX} = Z''_{YX}.$
\end{cor}
\begin{proof}
This is an immediate consequence of the following equivalence:
\begin{align*}
&\phantom{\iff i}\mathsf{F}(i_Y(y))\leq_{\MX} i_X(x)\\ 
&\iff \bv i_X\circ e_X[(i_Y\circ e_Y)^{-1}(i_Y(y)^\downarrow)] \leq_{\MX} \bw i_X\circ e_X[(i_X\circ e_X)^{-1}(i_X(x)^\uparrow)] \\
&\iff (\forall p\in e_Y^{-1}(y^\downarrow))(\forall q\in e_X^{-1}(x^\uparrow))\Big(i_X\circ e_X(p) \leq_{\MX} i_X\circ e_X(q)\Big)\\
&\iff (\forall p\in e_Y^{-1}(y^\downarrow))(\forall q\in e_X^{-1}(x^\uparrow))\Big(p\leq q\Big).
\end{align*}
\end{proof}

Corollary \ref{C:Galois} justifies the terminology `Galois polarity', as the unique 3-preorder for any Galois polarity $\Pe$ is defined by $\R$, the orders on $X$ and $Y$, and the Galois connection given by $\mathsf{F}$ and $\mathsf{G}$ for any suitable choices of $i_X$ and $i_Y$.

\subsection{Polarity morphisms}\label{S:polHom}

Recall from Definition \ref{D:ext} that an extension polarity $\Pe$ is complete if $e_X$ and $e_Y$ are completions. Noting Proposition \ref{P:GalSimp}, we see that \cite[Theorem 3.4]{GJP13} establishes a one-to-one correspondence between what we call complete Galois polarities and $\Delta_1$-completions of a poset. Theorem \ref{T:delta} below expands on the proof of this result, and in Section \ref{S:cat} we reformulate it in terms of an adjunction between categories. First we need to define a concept of morphism between Galois polarities.

\begin{defn}\label{D:polHom}
Let $P$ and $P'$ be posets, let $G=\Pe$ and $G'=(e_{X'},e_{Y'},\R')$ be Galois polarities extending $P$ and $P'$ respectively. Then a \textbf{polarity morphism} $h:G\to G'$ is a triple of order preserving maps 
\[h=(h_X:X\to X', h_P:P\to P', h_Y:Y\to Y')\] 
such that:
\begin{description}
\item[(M1)\label{M1}] The diagram in Figure \ref{F:isom} commutes.
\item[(M2)\label{M2}] For all $x\in X$ and $y\in Y$ we have 
\[\iota_Y(y)\leq \iota_X(x)\ra \iota_{Y'}\circ h_Y(y)\leq \iota_{X'}\circ h_X(x).\]
\item[(M3)\label{M3}] For all $x'\in X'$ and for all $y'\in Y'$, if $(x', y')\notin \R'$, then there is $x\in X$ and $y\in Y$ such that: 
\begin{enumerate}[(i)]
\item $h^{-1}_X(x'^\uparrow)\subseteq x^\uparrow$. 
\item $h_Y^{-1}(y'^\downarrow)\subseteq y^\downarrow$. 
\item $h_X(a) \R' y' \ra a \R y$ for all $a\in X$.
\item $x' \R' h_Y(b) \ra x\R b$ for all $b\in Y$. 
\item $(x,y)\notin \R$.
\end{enumerate}
\end{description}
For a polarity morphism $h$, if $h_X$, $h_P$ and $h_Y$ are all order embeddings, and also $h_X(x)\R' h_Y(y)\ra x\R y$ for all $x\in X$ and $y\in Y$, then $h$ is a \textbf{polarity embedding}. If, in addition, all maps are actually order isomorphisms, then $h$ is a \textbf{polarity isomorphism}, and we say $G$ and $G'$ are isomorphic. 

Sometimes we want to fix a poset $P$ and deal exclusively with isomorphism classes of Galois polarities extending $P$. In this case we say Galois polarities $G$ and $G'$ are \textbf{isomorphic as Galois polarities extending $P$} if there is a polarity isomorphism $(h_X,h_P,h_Y):G\to G'$ where $h_P$ is the identity on $P$. 
\end{defn}

\begin{figure}[htbp]
\[\xymatrix{ 
X\ar[d]_{h_X} & \ar[l]_{e_X}P\ar[r]^{e_Y}\ar[d]^{h_P} & Y\ar[d]^{h_Y}\\
X' & \ar[l]^{e_{X'}}P'\ar[r]_{e_{Y'}} & Y'
}\] 
\caption{}
\label{F:isom}
\end{figure}

Note that if $h_X$ and $h_Y$ are order embeddings, then $h_P$ will be too, but it is not the case that $h_X$ and $h_Y$ being order isomorphisms implies that $h_P$ is too, as $h_P$ may not be surjective. Note also that Definition \ref{D:polHom}, while being similar in some respects, is largely distinct from the notion of a \emph{bounded morphism between polarity frames} from \cite{Suz12}. It is also completely different to the frame morphisms of \cite{DGP05, Geh06}, which are duals to complete lattice homomorphisms, rather than `decomposed' versions of certain maps $X\uplus Y\to X'\uplus Y'$. We will make this clear in Theorem \ref{T:polHom} later. 

\begin{lemma}\label{L:xRy}
If $h= (h_X:X\to X', h_P:P\to P', h_Y:Y\to Y')$ is a polarity morphism, then for all $x\in X$ and for all $y\in Y$ we have $x\R y\ra h_X(x) \R' h_Y(y)$.
\end{lemma}
\begin{proof}
Suppose $(h_X(x),h_Y(y))\notin \R'$. Then, by \ref{M3} there are $x_0\in X$ and $y_0\in Y$ with $h_X^{-1}(h_X(x)^\uparrow)\subseteq x_0^\uparrow$, with $h_Y^{-1}(h_Y(y)^\downarrow)\subseteq y_0^\downarrow$ and with $(x_0,y_0)\notin \R$. From $h_X^{-1}(h_X(x)^\uparrow)\subseteq x_0^\uparrow$ it follows that $x_0\leq_X x$, and similarly we have $y\leq_Y y_0$.  Thus $(x,y)\notin \R$, as otherwise \ref{C1} and \ref{C2} would force $x_0\R y_0$.
\end{proof}

The following definition is due to Ern\'e. This will be useful to us as it precisely characterizes those maps between posets that lift (uniquely) to complete homomorphisms between their MacNeille completions \cite[Theorem 3.1]{Ern91a}.

\begin{defn}\label{D:cut}
An order preserving map $f:P\to Q$ is \textbf{cut-stable} if whenever $q_1\not\leq q_2\in Q$, there are $p_1\not\leq p_2\in P$ such that $f^{-1}(q_1^\uparrow)\subseteq p_1^\uparrow$ and $f^{-1}(q_2^\downarrow)\subseteq p_2^\downarrow$.
\end{defn}

\ref{M3} is related to cut-stability, as we shall see in Theorems \ref{T:polHom} and \ref{T:GalStab}. We can think of this as an adaptation of ideas from \cite[Section 4]{GehPri08}. We extend from what, according to our terminology, is the special case of $\Pe$ where $e_X$ and $e_Y$ are, respectively, the free directed meet- and join-completions and $\R=\R_l$, to Galois polarities in general. We will need the following definition.

\begin{defn}[$\psi_h$]\label{D:psi}
Let $G=\Pe$ and $G'=(e_{X'},e_{Y'},\R')$ be Galois polarities extending posets $P$ and $P'$ respectively. Let $\gamma:P\to X\uplus Y$ and $\gamma':P'\to X'\uplus Y'$ be the maps from Definition \ref{D:gamma}. Let $h=(h_X,h_P,h_Y):G\to G'$ be a polarity morphism. Define $\psi_h: X\uplus Y \to X'\uplus Y'$ by
\[\psi_h(z) = \begin{cases}\iota_{X'}\circ h_X(z) \text{ when } z\in \iota_X[X] \\ \iota_{Y'}\circ h_Y(z) \text{ when } z\in \iota_Y[Y]\end{cases}\]
\end{defn}
We prove that $\psi_h$ is well defined as part of Theorem \ref{T:polHom}, below. We will also use the following definition.

\begin{defn}\label{D:GalStab}
Let $G=\Pe$ and $G'=(e_{X'},e_{Y'},\R')$ be Galois polarities, let $\gamma:P\to X\uplus Y$ and $\gamma':P'\to X'\uplus Y'$ be as in Definition \ref{D:gamma}, and let $\psi: X\uplus Y \to X'\uplus Y'$. We say $\psi$ is \textbf{Galois-stable} if it is order preserving, cut-stable, and satisfies:
\begin{description}
\item[(G1)\label{G1}] $\psi\circ\gamma[P]\subseteq \gamma'[P']$.
\item[(G2)\label{G2}] $\psi\circ\iota_X[X]\subseteq \iota_{X'}[X']$.
\item[(G3)\label{G3}] $\psi\circ\iota_Y[Y]\subseteq \iota_{Y'}[Y']$.
\end{description} 
\end{defn}

\begin{thm}\label{T:polHom}
Let $G$, $G'$ $\gamma$ and $\gamma'$ be as in Definition \ref{D:GalStab}. Then, given a polarity morphism $h=(h_X, h_P,h_Y):G\to G'$, we have:
\begin{enumerate} 
\item $\psi_h$ is the unique map such that the diagram in Figure \ref{F:polHom} commutes (if we replace $\psi$ with $\psi_h$). 
\item $\psi_h$ is Galois-stable.
\item $\psi_h$ is an order embedding if and only if $h$ is a polarity embedding.
\item If $h_X$ and $h_Y$ are both surjective, then $\psi_h$ is surjective.
\end{enumerate}
\end{thm}
\begin{proof}
Let $\preceq$ and $\preceq'$ be the unique 3-preorders for $G$ and $G'$ respectively.
\begin{enumerate}[]
\item[(1)] Given $(h_X, h_P,h_Y)$, the commutativity of the diagram in Figure \ref{F:polHom} demands that $\psi_h$ can only be defined as in Definition \ref{D:psi}, if it exists at all. This deals with uniqueness. Now, if $x\in X$ and $y\in Y$, then, using Lemma \ref{L:xRy}, we have 
\[x\preceq y \iff x\R y\implies h_X(x)\R' h_Y(y)\iff h_X(x)\preceq' h_Y(y).\]

If $y\preceq x$, then $\iota_Y(y)\leq \iota_X(x)$ by definition, and so  $h_Y(y)\preceq' h_X(x)$ by \ref{M2}. This shows $\psi_h$ is well defined, and combined with the fact that $h_X$ and $h_Y$ are both order preserving, also shows $\psi_h$ is order preserving.

\item[(2)] We have already proved that $\psi_h$ is order preserving. To see that it is cut-stable, let $z_1\not\leq z_2\in X'\uplus Y'$. Since by Theorem \ref{T:unique}(3) we know $\iota_{X'}[X']$ and $\iota_{Y'}[Y']$ are, respectively, join- and meet-dense in $X'\uplus Y'$, there are $x'\in X'$ and $y'\in Y'$ with $\iota_{X'}(x')\leq z_1$, with $z_2\leq \iota_{Y'}(y')$, and with $\iota_{X'}(x')\not\leq \iota_{Y'}(y')$ (i.e. $(x',y')\notin \R'$). Thus by \ref{M3} there are $x\in X$ and $y\in Y$ with the five properties described in that definition. We will satisfy the condition of Definition \ref{D:cut} using the pair $\iota_X(x)\not\leq \iota_Y(y)$. 

Let $z\in \psi_h^{-1}(z_1^\uparrow)$. We must show that $z\in \iota_X(x)^\uparrow$. We have $\psi_h(z)\geq z_1\geq\iota_{X'}(x')$.  There are two cases. If $z= \iota_X(a)$ for some $a\in X$, then $\psi_h(z)= \iota_{X'}\circ h_X(a)$, and so $h_X(a)\geq_{X'} x'$. Thus $a\in h_X^{-1}(x'^\uparrow)$, and so $a\in x^\uparrow$, by \ref{M3}(i). It follows that $z = \iota_X(a)\in \iota_X(x)^\uparrow$ as claimed. Alternatively, suppose $z=\iota_Y(b)$ for some $b\in Y$. Then $\psi_h(z)=\iota_{Y'}\circ h_Y(b)$, and so $\iota_{Y'}\circ h_Y(b)\geq \iota_{X'}(x')$, and consequently $x' \R' h_{Y}(b)$. It follows from \ref{M3}(iv) that $x\R b$, and thus that $\iota_X(x)\leq \iota_Y(b)=z$ as required. That $\psi_h^{-1}(z_2^\downarrow)\subseteq \iota_Y(y)^\downarrow$ follows by a dual argument, and so $\psi_h$ is cut-stable.

To see that \ref{G1} holds for $\psi_h$ note that
\begin{align*}
\psi_h\circ\gamma(p) &= \psi_h\circ\iota_Y\circ e_Y(p) \\
&= \iota_{Y'}\circ h_Y\circ e_Y(p) \\
&= \iota_{Y'}\circ e_{Y'}\circ h_P(p)\\
&= \gamma'(h_P(p)).
\end{align*} 
That \ref{G2} and \ref{G3} hold follows immediately from the definition of $\psi_h$.

\item[(3)] If $\psi_h$ is an order embedding, then that $h_X$ and $h_Y$, and thus also $h_P$, are order embeddings follows directly from the commutativity of the diagram in Figure \ref{F:polHom} (substituting $\psi_h$ for $\psi$). Moreover, that $h_X(x)\R' h_Y(y)\ra x\R y$ for all $x\in X$ and $y\in Y$ holds for the same reason. Thus $(h_X,h_P,h_Y)$ is a polarity embedding. 

Conversely, suppose $(h_X,h_P,h_Y)$ is a polarity embedding. Since we already know $\psi_h$ is order preserving, suppose $z,z'\in X\uplus Y$ and that $\psi_h(z)\leq \psi_h(z')$. There are four cases.
If either $z,z'\in \iota_X[X]$, or $z,z'\in\iota_Y[Y]$, then that $z\leq z'$ follows again from the commutativity of the diagram in Figure \ref{F:polHom} and the assumption that $h_X$ and $h_Y$ are order embeddings. In the case where $z = \iota_X(x)$ and $z'=\iota_Y(y)$ for some $x\in X$ and $y\in Y$, we have 
\[\psi_h(z)\leq \psi_h(z')\iff h_X(x) \R' h_Y(y) \iff x\R y\iff \iota_X(x)\leq \iota_Y(y)\iff z\leq z'.\]
In the final case we have $z=\iota_Y(y)$ and $z'=\iota_X(x)$ for some $x\in X$ and $y\in Y$. Then
\[\psi_h(z)\leq \psi_h(z')\iff \iota_{Y'}\circ h_Y(y) \leq \iota_{X'}\circ h_X(x).\]
If $p,q\in P$, and $e_Y(p)\leq_Y y$ and $x\leq_X e_X(q)$, then
\[\iota_{Y'}\circ h_Y\circ e_Y(p)\leq \iota_{Y'}\circ h_Y(y)\leq \iota_{X'}\circ h_X(x)\leq \iota_{X'}\circ h_X\circ e_X(q),\]
and thus $\iota_{Y'}\circ e_{Y'}\circ h_P(p)\leq \iota_{X'}\circ e_{X'}\circ h_P(q)$, by the commutativity of the diagram in Figure \ref{F:isom}. It follows that $e_{Y'}\circ h_P(p)\preceq' e_{X'}\circ h_P(q)$, and thus that $h_P(p)\leq_{P'} h_P(q)$, as $\preceq'$ is a 3-preorder for $G'$. Thus $p\leq_P q$, as $h_P$ is an order embedding. So by Corollary \ref{C:alt} we have $y\preceq x$ as required.

\item[(4)] If $h_X$ and $h_Y$ are onto, then, given $z'\in X'\uplus Y'$, we have either $z' = \iota_{X'}(h_X(x))$ for some $x\in X$, or $z' = \iota_{Y'}(h_Y(y))$ for some $y\in Y$. In either case it follows there is $z\in X\uplus Y$ with $\psi_h(z)= z$, and thus that $\psi_h$ is onto. 
\end{enumerate}
\end{proof}
Note that the converse to (4) in the above theorem may not hold, as illustrated by Example \ref{E:notSurj}.
\begin{figure}[htbp]
\[\xymatrix{ 
& Y\ar[d]^{\iota_Y}\ar[r]^{h_Y} & Y'\ar[d]_{\iota_{Y'}}\\
P\ar[ur]^{e_Y}\ar[dr]_{e_X}\ar[r]^\gamma\ar@/_/[rrr]_{h_P} & X\uplus Y\ar[r]^{\psi} & X'\uplus Y' & \ar[dl]^{e_{X'}}\ar[ul]_{e_{Y'}}\ar[l]_{\gamma'}P' \\
& X\ar[u]_{\iota_X}\ar[r]_{h_X} & X'\ar[u]^{\iota_{X'}}
}\] 
\caption{}
\label{F:polHom}
\end{figure}

\begin{ex}\label{E:notSurj}
Let $P= X$ be a two element antichain, and let $Y$ be this two element antichain extended by adding a join for the two base elements. Let $P'\cong X'\cong Y' \cong Y$. Then the inclusion maps (modulo isomorphism) and the relations $\R_l$ and $\R'_l$ define Galois polarities $G$ and $G'$, and $X\uplus Y\cong Y\cong X'\uplus Y'$. Let $h_X$, $h_P$ and $h_Y$ be the maps induced by the inclusions $P\subset P'$, $X\subset X'$ and $Y = Y'$ (modulo isomorphism). Then $\psi_h$ is an isomorphism, and so is surjective, but $h_X$ is obviously not surjective.   
\end{ex}

\begin{defn}\label{D:h}
Let $G=\Pe$ and $G'=(e_{X'},e_{Y'},\R')$ be Galois polarities. Let $\gamma:P\to X\uplus Y$ and $\gamma':P'\to X'\uplus Y'$ be as in Definition \ref{D:gamma}. Let $\psi: X\uplus Y\to X'\uplus Y'$ satisfy \ref{G1}--\ref{G3}. Define $h^\psi_X:X\to X'$, $h^\psi_Y:Y\to Y'$ and $h^\psi_P:P\to P'$ as follows:
\begin{align*}
h^\psi_X =& \iota_{X'}^{-1}\circ \psi\circ \iota_X \\
h^\psi_Y=& \iota_{Y'}^{-1}\circ \psi\circ \iota_Y \\
h^\psi_P =& \gamma'^{-1}\circ \psi\circ\gamma. 
\end{align*}
Here $\iota_{X'}^{-1}$, $\iota_{Y'}^{-1}$ and $\gamma'^{-1}$ are the partial inverse maps. Define $h_\psi = (h^\psi_X,h^\psi_P,h^\psi_Y)$.
\end{defn}

The maps $\iota_{X'}^{-1}$, $\iota_{Y'}^{-1}$ and $\gamma'^{-1}$ are well defined as partial functions because $\iota_{X'}$, $\iota_{Y'}$ and $\gamma'$ are all order embeddings. That $h^\psi_X$, $h^\psi_Y$ and $h^\psi_P$ are well defined will be proved as part of the following theorem.

\begin{thm}\label{T:GalStab}
Let $G$, $G'$, $\gamma$ and $\gamma'$ be as in Definition \ref{D:h}. Let $\psi: X\uplus Y\to X'\uplus Y'$  be a Galois-stable map, and let $h^\psi_X$, $h^\psi_Y$ and $h^\psi_P$ be as above. Then:
\begin{enumerate}
\item $h^\psi_X$, $h^\psi_Y$ and $h^\psi_P$ are the unique maps such that the diagram in Figure \ref{F:polHom} commutes (if we replace $h_X, h_Y$ and $h_P$ with $h^\psi_X$, $h^\psi_Y$ and $h^\psi_P$, respectively).
\item $h_\psi = (h^\psi_X, h^\psi_P,h^\psi_Y)$ is a polarity morphism.
\end{enumerate}
\end{thm}
\begin{proof}\mbox{}
\begin{enumerate}[]
\item[(1)] If the diagram in Figure \ref{F:polHom} is to commute, then $h^\psi_X$ and $h^\psi_Y$ must be $\iota_{X'}^{-1}\circ \psi\circ \iota_X$ and $\iota_{Y'}^{-1}\circ \psi\circ \iota_Y$ respectively, if they are to exist at all. It follows from \ref{G2} that $\iota_{X'}^{-1}$ is total on $\psi\circ \iota_X[X]$, and thus that $h^\psi_X$ is well defined. A similar argument with \ref{G3} shows $h^\psi_Y$ is also well defined. The commutativity of this diagram also demands that, if $h^\psi_P$ exists, we have 
\[e_{X'}\circ h^\psi_P = h^\psi_X\circ e_X = \iota_{X'}^{-1}\circ \psi \circ \iota_X\circ e_X = \iota_{X'}^{-1}\circ \psi \circ\gamma,\] 
and thus 
\[h^\psi_P=e_{X'}^{-1}\circ \iota_{X'}^{-1}\circ \psi \circ\gamma=\gamma'^{-1}\circ \psi\circ\gamma,\] if this is well defined, which it is, by \ref{G1}. 

\item[(2)]To see that that $h_\psi$ satisfies \ref{M1}, observe that we have
\begin{align*}
h^\psi_X\circ e_X &= \iota_{X'}^{-1}\circ \psi\circ \iota_X \circ e_X\\
&= \iota_{X'}^{-1}\circ\psi\circ\gamma \\
&= \iota_{X'}^{-1}\circ \gamma'\circ h^\psi_P\\
& = \iota_{X'}^{-1}\circ \iota_{X'}\circ e_{X'}\circ h^\psi_P\\
& = e_{X'}\circ h^\psi_P,
\end{align*} 
and $h^\psi_Y\circ e_Y =  e_{Y'}\circ h^\psi_P$ by a similar argument. That $(h^\psi_X, h^\psi_P, h^\psi_Y)$ satisfies \ref{M2} also follows easily from the definitions of $h_X$ and $h_Y$ combined with the fact that $\psi$ is order preserving. 

It remains only to check \ref{M3}. As $\psi$ is Galois-stable it is necessarily cut-stable. So, let $x'\in X'$, let $y'\in Y'$, and suppose $(x', y')\notin \R'$. Then $\iota_{X'}(x')\not\leq \iota_{Y'}(y')$, and thus by cut-stability there are $z_1\not\leq z_2\in X\uplus Y$ with $\psi^{-1}(\iota_{X'}(x')^\uparrow)\subseteq z_1^\uparrow$, and $\psi^{-1}(\iota_{Y'}(y')^\downarrow)\subseteq z_2^\downarrow$. As $\iota_X[X]$ and $\iota_Y[Y]$ are, respectively, join- and meet-dense in $X\uplus Y$, there are $x\in X$ and $y\in Y$ with $\iota_X(x)\leq z_1$, with $z_2\leq \iota_Y(y)$, and with $(x,y)\not\in \R$. It follows that $\psi^{-1}(\iota_{X'}(x')^\uparrow)\subseteq \iota_X(x)^\uparrow$ and $\psi^{-1}(\iota_{Y'}(y')^\downarrow)\subseteq \iota_Y(y)^\downarrow$. We check the conditions required for \ref{M3} are satisfied by the pair $(x,y)$:
\begin{enumerate}[(i)]
\item Let $a\in X$ and suppose $a\in h_X^{\psi-1}(x'^\uparrow)$. Then $h^\psi_X(a)\geq_{X'} x'$ and so, by definition of $h^\psi_X$, we have $\iota_X(a)\in \psi^{-1}(\iota_{X'}(x')^\uparrow)\subseteq \iota_X(x)^\uparrow$, and so $a\in x^\uparrow$.
\item Dual to (i).
\item Let $a\in X$ and suppose $h^\psi_X(a)\R' y'$. Then $\iota_{X'}\circ h^\psi_X(a) \leq \iota_{Y'}(y')$, and thus $\psi \circ \iota_X(a) \leq \iota_{Y'}(y')$. It follows that $\iota_X(a)\in \psi^{-1}(\iota_{Y'}(y')^\downarrow)\subseteq \iota_Y(y)^\downarrow$, and so $a\R y$ as required.
\item Dual to (iii). 
\item By choice of $(x,y)$.
\end{enumerate}   
\end{enumerate}
\end{proof}

\begin{thm}\label{T:hom-stab}
Let $G=\Pe$ and $G'=(e_{X'},e_{Y'},\R')$ be Galois polarities. Let $h:G\to G'$ be a polarity morphism. Then, recalling Definitions \ref{D:psi} and \ref{D:h}, we have 
\[h = h_{\psi_h}.\] 
Similarly, if $\psi:X\uplus Y\to X'\uplus Y'$ is a Galois-stable map, then 
\[\psi = \psi_{h_\psi}.\]
\end{thm}
\begin{proof}
Given $h$, by Theorem \ref{T:polHom} we know $\psi_h$ is the unique map such that the diagram in Figure \ref{F:polHom} commutes (replacing $\psi$ with $\psi_h$). We also know that $\psi_h$ is Galois-stable. But by Theorem \ref{T:GalStab} it follows that $h_{\psi_h}=(h^{\psi_h}_X,h^{\psi_h}_P,h^{\psi_h}_Y)$ is the unique triple such that the same diagram commutes (replacing $\psi$ with $\psi_h$, and replacing $h_X$, $h_P$ and $h_Y$ $h^{\psi_h}_X$, $h^{\psi_h}_P$ and $h^{\psi_h}_Y$ respectively). From this it follows immediately that $h = h_{\psi_h}$, as we now have two `unique' commuting diagrams involving $\psi_h$. The argument proving $\psi = \psi_{h_\psi}$ is essentially the same.
\end{proof}
We can now easily prove a result analogous to Theorem \ref{T:polHom}(3).

\begin{cor}
Let $G$, $G'$, $\psi$ and $h_\psi$ be as in Theorem \ref{T:GalStab}. Then $\psi$ is an order embedding if and only if $h_\psi$ is a polarity embedding. 
\end{cor}
\begin{proof}
We have $\psi = \psi_{h_\psi}$ from Theorem \ref{T:hom-stab}, and from Theorem \ref{T:polHom}(3) we have that $\psi_{h_\psi}$ is an order embedding if and only if $h_\psi$ is a polarity embedding. The result follows immediately.
\end{proof}

The next lemma is needed to show polarity morphisms compose as expected.

\begin{lemma}\label{L:compositions}
For each $i\in\{1,2,3\}$, let $G_i = (e_{X_i},e_{X_i},\R_i)$ be a Galois polarity, and let $\psi_1:X_1\uplus Y_1\to X_2\uplus Y_2$ and $\psi_2:X_2\uplus Y_2 \to X_3\uplus Y_3$ be Galois-stable maps. Then the composition $\psi_2\circ \psi_1$ is also Galois-stable, and $h_{\psi_2\circ \psi_1} = h_{\psi_2}\circ h_{\psi_1}$.

Similarly, let $h_1 = (h_{X_1},h_{P_1},h_{Y_1}):G_1\to G_2$ and $h_2=(h_{X_2},h_{P_2},h_{Y_2}):G_2\to G_3$ be polarity morphisms. If we define
\[h_2\circ h_1= (h_{X_2}\circ h_{X_1}, h_{P_2}\circ h_{P_1},h_{Y_2}\circ h_{Y_1}),\]
then $h_2\circ h_1$ is also a polarity morphism, and $\psi_{h_1\circ h_2} = \psi_{h_2}\circ \psi_{h_1}$.
\end{lemma}
\begin{proof}
It's straightforward to show that the composition of maps satisfying \ref{G1}--\ref{G3} also satisfies these conditions, and compositions of order preserving maps are obviously order preserving. Moreover, cut-stability is preserved by composition \cite[Corollary 2.10]{Ern91a}. So compositions of Galois-stable maps are also Galois-stable. Moreover, unpacking Definition \ref{D:h} reveals that $h_{\psi_2\circ \psi_1} = h_{\psi_2}\circ h_{\psi_1}$. For example, 
\begin{align*} h^{\psi_2\circ \psi_1}_{X_1}&=\iota_{X_3}^{-1}\circ \psi_2\circ\psi_1\circ\iota_{X_1}\\
& = \iota_{X_3}^{-1}\circ \psi_2\circ \iota_{X_2}\circ\iota_{X_2}^{-1}\circ\psi_1\circ\iota_{X_1}\\
&=h^{\psi_2}_{X_1}\circ h^{\psi_1}_{X_1}.\end{align*}

Similarly, unpacking Definition \ref{D:psi} reveals that $\psi_{h_2\circ h_1} = \psi_{h_2}\circ \psi_{h_1}$, and is consequently well defined and Galois-stable, and so it follows from Theorem \ref{T:GalStab} that  $h_{\psi_{h_2}\circ \psi_{h_1}}$ is a polarity morphism. So $h_2\circ h_1$ is a polarity morphism, as, using Theorem \ref{T:hom-stab}, we have 
\begin{align*}
h_2\circ h_1 &= h_{\psi_{h_2}}\circ h_{\psi_{h_1}} \\
&= h_{\psi_{h_2}\circ \psi_{h_1}}.
\end{align*}

\end{proof}

\begin{prop}
The class of Galois polarities and polarity morphisms forms a category, where polarity morphisms compose componentwise.
\end{prop}
\begin{proof}
Identity morphisms obviously exist, so we need only check composition behaves appropriately, and this follows immediately from Lemma \ref{L:compositions}.
\end{proof}

We will expand on this categorical viewpoint in Section \ref{S:cat}.

\subsection{Galois polarities and $\Delta_1$-completions}\label{S:PolAndDel}
Recall that if $P$ is a poset we write, for example, $e:P\to \DM(P)$ for the MacNeille completion of $P$ (see Definition \ref{D:DM}).

\begin{lemma}\label{L:polToDel}
If $\Pe$ is a Galois polarity, and if $e:X\uplus Y\to\DM(X\uplus Y)$ is the MacNeille completion of $X\uplus Y$, then $e\circ\gamma:P\to \DM(X\uplus Y)$ is a $\Delta_1$-completion (where $\gamma$ is as in Definition \ref{D:gamma}). 
\end{lemma}
\begin{proof}
$X\uplus Y$ is join-generated by $\iota_X[X]$, and meet-generated by $\iota_Y[Y]$ (by Theorem \ref{T:unique}(3)), and $X$ and $Y$ are meet- and join-generated by $e_X[P]$ and $e_Y[P]$ respectively. $\DM(X\uplus Y)$ is both join- and meet-generated by $e[X\uplus Y]$. Thus every element of $\DM(X\uplus Y)$ is both a join of meets, and a meet of joins, of elements of $e\circ\gamma[P]$ as required.
\end{proof}

\begin{defn}\label{D:polToDel}
If $G$ is a Galois polarity, define the $\Delta_1$-completion $d_G=e\circ\gamma$ constructed from $G$ as in Lemma \ref{L:polToDel} to be the \textbf{$\Delta_1$-completion generated by $G$}.
\end{defn}

Now, given a $\Delta_1$-completion $d:P\to D$, define $X_D$ and $Y_D$ to be (disjoint isomorphic copies of) the subsets of $D$ meet- and join-generated by $d[P]$ respectively. Define $e_{X_D}:P\to X_D$ and $e_{Y_D}:P\to Y_D$ by composing $d$ with the isomorphisms into $X_D$ and $Y_D$ respectively. Abusing notation by identifying $X_D$ and $Y_D$ with their images in $D$, define $\R_D$ on $X_D\times Y_D$ by $x \R_D y\iff x\leq y$ in $D$. Then $(e_{X_D}, e_{Y_D},\R_D)$ is a complete Galois polarity. To see this note that it is straightforward to show that the order on $D$ induces an order on $X_D\cup Y_D$ satisfying the necessary conditions. This leads to the following definition.

\begin{defn}\label{D:delToPol}
If $d:P\to D$ is a $\Delta_1$-completion, define the complete Galois polarity $G_d=(e_{X_D}, e_{Y_D},\R_D)$  constructed from $d$ as described in the preceding paragraph to be the \textbf{Galois polarity generated by $d$}.
\end{defn}

\begin{lemma}\label{L:GalEmb}
Let $G=\Pe$ be a Galois polarity. Let $d_G=e\circ\gamma:P\to \DM(X\uplus Y)$ be the $\Delta_1$-completion generated by $G$.  Denote the Galois polarity generated by $d_G$ by $G_{d_G}=(e_{X_{\DM}}, e_{Y_{\DM}}, \R_{\DM} )$. Then there is a polarity embedding $(h_X,h_P,h_Y)$ from $G$ to $G_{d_G}$ with $h_P = \id_P$. Moreover, if $G$ is complete, then this embedding is an isomorphism of polarities extending $P$, and is the only such isomorphism.
\end{lemma}
\begin{proof}
Recall that $\gamma = \iota_X\circ e_X = \iota_Y\circ e_Y$ by definition. Define 
$h_P=\id_P$. Note that we are using e.g. $X_{\DM}$ to denote $X_{\DM(X\uplus Y)}$. Let $\mu_X$ be the isomorphism used to define $X_{\DM}$ (recall that $X_{\DM}$ is an isomorphic copy of a subset of $\DM(X\uplus Y)$). Define the map $h_X:X\to X_{\DM}$ by $h_X = \mu_X\circ e\circ\iota_X$. This is clearly a meet-completion. Similarly define a meet-completion $h_Y:Y\to Y_{\DM}$ by $h_Y=\mu_Y\circ e\circ\iota_Y$. Note that $e_{X_{\DM}}$ is just $\mu_X\circ e\circ\gamma = \mu_X\circ e\circ\iota_X\circ e_X$, and similarly $e_{Y_{\DM}}=\mu_Y\circ e\circ\iota_Y\circ e_Y$. Thus we trivially have the commutativity required by \ref{M1}. 

To show that \ref{M2} is also satisfied, let $x\in X$, let $y\in Y$, and suppose $\iota_Y(y)\leq \iota_X(x)$. Then $e\circ \iota_Y(y)\leq e\circ \iota_X(x)$. The unique 3-preorder $\preceq = \hRg$ for $G_{d_G}$ can only be the order inherited from $\DM(X\uplus Y)$, so $\mu_X\circ e\circ\iota_X(x)\preceq \mu_Y\circ e\circ\iota_Y(y)$, and thus $\iota_{Y_{\DM}}\circ h_Y(y) \leq \iota_{X_{\DM}}\circ h_X(x)$ as required.

For \ref{M3}, let $x'\in X_{\DM}$, let $y'\in Y_{\DM}$, and suppose $(x', y')\notin \R_{\DM}$. Then $\mu^{-1}_X(x')\not\leq \mu^{-1}_Y(y')$. As $e\circ\iota_X[X]$ and $e\circ\iota_Y[Y]$ are, respectively, join- and meet-dense in $\DM(X\uplus Y)$, there is $x\in X$ and $y\in Y$ with $e\circ\iota_X(x)\not\leq e\circ\iota_Y(y)$, with  $e\circ\iota_X(x)\leq \mu^{-1}_X(x')$, and with $\mu^{-1}_Y(y')\leq e\circ\iota_Y(y)$. We check the necessary conditions are satisfied for this choice of $x$ and $y$:
\begin{enumerate}[(i)]
\item Let $a\in X$. Then 
\begin{align*} h_X(a)\geq_{X_\DM} x' &\iff \mu^{-1}_X\circ\mu_X\circ e\circ\iota_X(a)\geq \mu^{-1}_X(x')\\
&\implies e\circ\iota_X(a)\geq e\circ\iota_X(x)\\
&\iff a\geq_X x,
\end{align*}
and so $h_X^{-1}(x'^\uparrow)\subseteq x^\uparrow$ as required.
\item Dual to (i).
\item Let $a\in X$, let $y'\in Y_{\DM}$, and suppose $h_X(a)\R_{\DM} y'$. Then 
\[\mu^{-1}_X\circ\mu_X\circ e\circ\iota_X(a)\leq \mu_Y^{-1}(y')\leq e\circ\iota_Y(y),\]
and so $\iota_X(a)\leq \iota_Y(y)$, and thus $a\R y$ as required.
\item Dual to (iii).
\item Since $e\circ\iota_X(x)\not\leq e\circ\iota_Y(y)$ we must have $(x, y)\notin \R$.
\end{enumerate}
Now let $x\in X$, let $y\in Y$, and suppose $h_X(x)\R_{\DM} h_Y(y)$. Then 
\[\mu_X^{-1}\circ\mu_X\circ e\circ\iota_X(x)\leq \mu_Y^{-1}\circ\mu_Y\circ e\circ\iota_Y(y),\] 
and so $\iota_X(x)\leq \iota_Y(y)$, and thus $x\R y$, and we conclude that $(h_X, \id_P, h_Y)$ is a polarity embedding as claimed. Finally, suppose $e_X$ and $e_Y$ are completions. Then, given $x'\in X_\DM$, there is $S\subseteq P$ with
\begin{align*} x' &= \mu_X(\bw e\circ\gamma[S]) \\
&= \mu_X(\bw e\circ\iota_X\circ e_X[S])\\
&= \mu_X\circ e\circ\iota_X(\bw e_X[S])\\
&= h_X(\bw e_X[S]),
\end{align*} 
so $h_X$ is surjective, and thus an isomorphism. The same is true for $h_Y$ by duality, and, as $h_P$ is the identity on $P$, we have a polarity isomorphism extending $P$ as claimed. Finally, suppose $g = (g_X, \id_P, g_Y)$ is another such polarity isomorphism. Then we must have $g_X\circ e_X = e_{X_\DM} = h_X\circ e_X$. Since $e_X$ is a meet-extension it follows that $g_X = h_X$. A dual argument gives $g_Y = h_Y$, so we are done. 
\end{proof}

\begin{lemma}\label{L:delIsom}
Let $d:P\to D$ be a $\Delta_1$-completion. Let $G_d=(e_{X_D}, e_{Y_D},\R_D)$ be the complete Galois polarity generated by $d$.  Then $d$ is isomorphic to $d_{G_d}$, the $\Delta_1$-completion generated by $G_d$.  
\end{lemma}
\begin{proof}
Let $\mu_X$ and $\mu_Y$ be the isomorphisms onto $X_D$ and $Y_D$, respectively, let $\gamma_d:P\to X_D\uplus Y_D$, and let $e_d: X_D\uplus Y_D\to \DM( X_D\uplus Y_D)$ be the MacNeille completion of $ X_D\uplus Y_D$. Then $\mu_X^{-1}(X)\cup \mu_Y^{-1}(Y)$ has an order inherited from $D$, and we have the situation described in Figure \ref{F:delIsom}. This gives the result, as $e_d\circ\gamma_d$ is the $\Delta_1$-completion generated by $G_d$. 
\end{proof}

\begin{figure}[htbp]
\[\xymatrix{ 
P\ar[r]^{\gamma_d}\ar[dr]_d& X_D\uplus Y_D\ar[r]^{e_d}\ar@{<->}[d]^\cong & \DM(X_D\uplus Y_D)\ar@{<->}[d]^\cong\\
& \mu_X^{-1}(X)\cup \mu_Y^{-1}(Y)\ar[r]_{\phantom{xxxx}\subseteq} & D
}\] 
\caption{}
\label{F:delIsom}
\end{figure}

\begin{thm}\label{T:delta}
There is a 1-1 correspondence between (isomorphism classes of) $\Delta_1$-completions and (isomorphism classes of) complete Galois polarities. Moreover, for a fixed poset $P$ this correspondence restricts to a 1-1 correspondence between (isomorphism classes of) $\Delta_1$-completions of $P$ and (isomorphism classes of) complete Galois polarities extending $P$. 
\end{thm}
\begin{proof}
We define a map $\Theta$ from the class of all isomorphism classes of $\Delta_1$-completions to the class of all isomorphism classes of complete Galois polarities. Given a $\Delta_1$-completion $d$, let $[d]$ be the associated isomorphism class, and define $\Theta([d])$ by $G\in \Theta([d])$ if and only if $G$ is isomorphic to $G_d$, the complete Galois polarity generated by $d$.

Now, if $G\in \Theta([d])$ and  $G'\cong G$, then $G'\in \Theta([d])$ by definition of $\Theta$. Conversely, if $G$ and $G'$ are both in $\Theta([d])$, then they are both isomorphic to $G_d$, and thus each other.  So $\Theta([d])$ is indeed an isomorphism class, provided it is well defined. Suppose then that $d\cong d'$, let $G\in \Theta([d])$ and let $G'\in \Theta([d'])$. Since $d$ and $d'$ are isomorphic, it's easy to construct an isomorphism between the complete Galois polarities they generate (since these polarities are based on subsets of the corresponding complete lattices). If we denote these generated Galois polarities by $G_d$ and $G_{d'}$,  we have $G\cong G_d\cong G_{d'}\cong G'$, and thus $G\in \Theta([d'])$ and $G'\in \Theta([d])$. It follows that $\Theta([d])=\Theta([d'])$, and so $\Theta$ is well defined.

By Lemma \ref{L:GalEmb}, if $G$ is a complete Galois polarity, then $G\cong G_{d_G}$, so $G \in \Theta([d_G])$, and so $\Theta$ is surjective. Now suppose that $\Theta([d])= G = \Theta([d'])$. Then $G_d\cong G\cong G_{d'}$. Now, it's easy to see that isomorphic Galois polarities generate isomorphic $\Delta_1$-completions (the map $\psi_h$ will be an isomorphism, by Theorem \ref{T:phi}, and this will lift to the required isomorphism), so 
\[d_{G_d}\cong d_G\cong d_{G_{d'}}.\] By Lemma \ref{L:delIsom} we have $d_{G_d}\cong d$ and $d_{G_{d'}}\cong d'$, and so $d\cong d'$. This shows $\Theta$ is injective.

Finally, appealing to Proposition \ref{P:GalSimp}, the second claim is essentially \cite[Theorem 3.4]{GJP13}, and we can also obtain this result with the proof above by working modulo $\Delta_1$-completions of a fixed poset $P$,  and complete Galois polarities extending $P$.
\end{proof}

Before moving on we pause to consider a technical question regarding the polarity extensions discussed in Section \ref{S:ExtRes}. Given a Galois polarity $\Pe$ which is not complete, by Lemma \ref{L:GalEmb} there is a polarity embedding from $\Pe$ to $(e_{X_D}, e_{Y_D}, \R_D)$, where $e_{X_D}$ and $e_{Y_D}$ are meet- and join-completions respectively. By the definition of polarity embeddings, there are order embeddings $h_X:X\to X_D$ and $h_Y:Y\to Y_D$, and its easy to see these will be meet- and join-completions respectively. 

Thus Theorem \ref{T:Galois} applies and produces a Galois polarity $(e_{X_D}, e_{Y_D}, \bR)$. We certainly have $\bR\subseteq \R_D$, by Theorem \ref{T:Galois}(5), but does the other inclusion also hold? The answer, in general, is no. To see this we borrow \cite[Example 2.2]{GJP13}, and lean heavily on the discussion at the start of Section 4 in that paper. The MacNeille completion of a poset $P$ can be constructed from the Galois polarity $(e_{\cF_p}, e_{\cI_p}, \R_l)$, where $e_{\cF_p}:P\to \cF_p$ and $e_{\cI_p}:P\to \cI_p$ are the natural embeddings into the sets of principal upsets and downsets of $P$ respectively. 

Consider the poset $P= \omega \cup \omega^\partial$. I.e. $P$ is made up of a copy of $\omega$ below a disjoint copy of the dual $\omega^\partial$. Then $\cN(P)=\omega\cup\{z\}\cup \omega^\partial$. I.e. $P$ with an additional element above $\omega$ and below $\omega^\partial$. Let $X = \cF_p$ and let $Y=\cI_p$, so $\cN(P)$ is generated by $(e_X, e_Y, \R_l)$. Let the complete Galois polarity arising from Theorem \ref{T:delta} be $(e_{X_D}, e_{Y_D}, \R_D)$, where $X_D \cong X\cup \{z_X\}$ (a copy of $\omega\cup\omega^\partial$ with $z_X$ in between), and $Y_D\cong Y\cup\{z_Y\}$ (a copy of $\omega\cup \omega^\partial$ $z_Y$ in between). Now, to produce $\cN(P)$ it is necessary that $z_X\R_D z_Y$, but $(z_X, z_Y)\notin \bR_l$, and thus $\R_D\neq \bR$ in this case.

It also follows from this that the $\Delta_1$-completion generated by $\Pe$ may not be isomorphic to the one generated by $(i_{X}\circ e_X, i_{Y}\circ e_Y, \bR)$ from Theorem \ref{T:Galois}, even when $i_{X}$ and $i_{Y}$ are meet- and join-completions respectively (just set $i_X = h_X$ and $i_Y = h_Y$ and use the example above).        

\subsection{A categorical perspective}\label{S:cat}

Here we assume some familiarity with the basic concepts of category theory. The standard reference is \cite{MacL98}, and an accessible introduction can be found in \cite{Lein14}.

\begin{defn}We define a pair of categories, $\pol$ and $\del$ as follows:
\begin{itemize}
\item[$\pol$:] Let $\pol$ be the category of Galois polarities and polarity morphisms (from Definition \ref{D:polHom}).
\item[$\del$:] Let $\del$ be the category whose objects are $\Delta_1$-completions, and whose maps are commuting squares as described in Definition \ref{D:mapCat}, with the additional property that $g_Q$ in that definition is a complete lattice homomorphism.
\end{itemize}
\end{defn}

\begin{defn}\label{D:FG}
Let $\Gamma:\pol \to \del$ and $\Delta :\del\to \pol$ be defined as follows:
\begin{itemize}
\item[$\Gamma$:] Let $\Gamma:\pol\to\del$ be the map that takes a Galois polarity $G$ to the $\Delta_1$-completion it generates (recall Definition \ref{D:polToDel}), and takes a polarity morphism $h:G\to G'$ to the map between $\Delta_1$-completions described in Figure \ref{F:Fmap}, where $\DM(\psi_h)$ is the unique complete lattice homomorphism lift of the map $\psi_h$ from Definition \ref{D:psi} to the respective MacNeille completions, as described in \cite[Theorem 3.1]{Ern91a}.  
\item[$\Delta $:] Let $\Delta :\del\to \pol$ be the map that takes a $\Delta_1$-completion $d:P\to D$ to the Galois polarity it generates (recall Definition \ref{D:delToPol}), and takes a $\Delta_1$-completion morphism as in Figure \ref{F:Dmap} to the triple $(g|_{X_{D_1}}, f, g|_{Y_{D_1}})$, where, for example, $g|_{X_{D_1}}$ is (modulo isomorphism) the restriction of $g$ to $X_{D_1}$. 
\end{itemize}
\end{defn}
That $\Gamma$ and $\Delta $ are well defined is proved as part of the next theorem.
\begin{thm}\label{T:adj} 
$\Gamma$ and $\Delta $ are functors and form an adjunction $\Gamma\dashv \Delta $.
\end{thm}
\begin{proof}For easier reading we will break the proof down into discrete statements which obviously add up to a proof of the claimed result.
\begin{enumerate}[$\bullet$]
\item ``\emph{$\Gamma$ is well defined}". $\Gamma$ is certainly well defined on objects. For maps, as Theorem \ref{T:polHom} says $\psi_h$ will be cut-stable, \cite[Theorem 3.1]{Ern91a} says that $\DM(\psi_h)$ will be a complete lattice homomorphism. Moreover, the conjunction of these theorems also implies that the diagram in Figure \ref{F:Fmap} commutes. 
\item ``\emph{$\Gamma$ is a functor}". To see that $\Gamma$ lifts identity maps to identity maps note that the identity morphism on $\Pe$ clearly lifts via Definition \ref{D:psi} to the identity on $X\uplus Y$, and since taking MacNeille completions is functorial for cut-stable maps (see \cite[Corollary 3.3]{Ern91a}), that $\Gamma$ maps identity morphisms appropriately follows immediately.  Similarly, given polarity morphisms $h_1$ and $h_2$, if $h_2\circ h_1$ exists we have $\psi_{h_2\circ h_1} = \psi_{h_2}\circ \psi_{h_1}$, by Lemma \ref{L:compositions}, so $\Gamma$ respects composition as the MacNeille completion functor does. 

\item ``\emph{$\Delta $ is well defined}". $\Delta $ is also clearly well defined on objects. Consider now a $\Delta_1$-completion map as in Figure \ref{F:Dmap}. Abusing notation by, for example, identifying $X_{D_1}$ with its isomorphic copy inside $D_1$, we must show that $g|_{X_{D_1}}:X_{D_1}\to X_{D_2}$, that $g|_{Y_{D1}}:Y_{D_1}\to Y_{D_2}$, and that $(g|_{X_{D_1}}, f, g|_{Y_{D1}})$ satisfies the conditions of Definition \ref{D:polHom}. First, to lighten the notation define $g_X = g|_{X_{D_1}}$, and $g_Y = g|_{Y_{D_1}}$. Now, let $x\in X_{D_1}$. Then there is $S\subseteq P_1$ with $x = \bw d_1[S]$. So, using the commutativity of the diagram in Figure \ref{F:Dmap} and the fact that $g$ is a complete lattice homomorphism, we have 
\[g(x) = g(\bw d_1[S]) = \bw g\circ d_1[S] = \bw d_2\circ f[S].\]
This proves that $g_{X_{D_1}}$ does indeed have codomain $X_{D_2}$, and the analogous result for $g_Y$ holds by duality.  Since $g_X$, $f$, and $g_Y$ are all obviously order preserving, it remains only to check the conditions of Definition \ref{D:polHom}:
\begin{enumerate}[]
\item \ref{M1}: This follows immediately from the definitions of $g_X$ and $g_Y$ and the commutativity of the diagram in Figure \ref{F:Dmap}.
\item \ref{M2}: Without loss of generality, we can think of the $\iota$ maps as inclusion functions, and so the claim is just the statement that $y\leq_{D_1} x\ra g(y)\leq_{D_2} g(x)$, and thus is true as $g$ is order preserving. 
\item \ref{M3}: Abusing notation in the same way as before, let $x'\in X_{D_2}$, let $y'\in Y_{D_2}$, and suppose $x'\not\leq y' \in D_2$. Using the completeness of $D_1$, let $z_1 = \bw g^{-1}(x'^\uparrow)$, and let $z_2 = \bv g^{-1}(y'^\downarrow)$. It follows easily from the fact that $g$ is a complete lattice homomorphism that $x'\leq_{D_2} g(z_1)$ and $g(z_2)\leq_{D_2} y'$, so if $z_1\leq_{D_1} z_2$, then $x'\leq_{D_2} g(z_1)\leq_{D_2} g(z_2) \leq_{D_2} y'$, as $g$ is order preserving. Thus to avoid contradiction we must have $z_1\not\leq_{D_1} z_2$. 

As $X_{D_1}$ and $Y_{D_1}$ are, respectively, join- and meet-dense in $D_1$, there is $x\in X_{D_1}\cap z_1^\downarrow$ and $ y\in Y_{D_1}\cap z_2^\uparrow$ with $x\not\leq_{D_1} y$. Now, as $x\leq_{D_1} z_1=\bw g^{-1}(x'^\uparrow)$ we have $g_X^{-1}(x'^\uparrow)\subseteq z_1^\uparrow\subseteq x^\uparrow$. Thus (i) holds for this choice of $x$, and (ii) holds for $y$ by a dual argument.

Moreover, suppose $a\in X_{D_1}$, and that $g(a)\leq_{D_2} y'$. Then $a\in g^{-1}(y'^\downarrow)$, and so $a\leq_{D_1} z_2$, by definition of $z_2$, and consequently $a\leq_{D_1} y$. Thus (iii) holds, and (iv) is dual. That (v) holds is automatic from the choice of $x$ and $y$.
\end{enumerate}
\item ``\emph{$\Delta $ is a functor}".  $\Delta $ obviously sends identity maps to identity maps, and almost as obviously respects composition.
\item ``$\Gamma\dashv \Delta $". The unit $\eta$ is defined so that its components are the embeddings of $G= \Pe$ into $\Delta \Gamma(G) = G_{d_D}=(e_{X_\DM},e_{Y_\DM},\R_\DM)$ described in Lemma \ref{L:GalEmb}. From this lemma, if we again abuse notation by treating $X_\DM$ and $Y_\DM$ as if they are subsets of $\DM(X\uplus Y)$, we have 
\[\eta_G = (e\circ\iota_X, \id_P,e\circ\iota_Y).\] We first show that $\eta$ is indeed a natural transformation. Let $G_1 = (e_{X_1}, e_{Y_1}, \R_1)$ and $G_2 = (e_{X_2}, e_{Y_2}, \R_2)$ be Galois polarities, and let $g =(g_X, g_P, g_Y)$ be a polarity morphism from $G_1$ to $G_2$. We aim to show that the diagram in Figure \ref{F:etaSquare} commutes. 

Define $h_{X_1} = e_1\circ \iota_{X_1}$ to be the standard embedding $X_1$ of into $X_{\DM_1}$, and similarly define $h_{Y_1} = e_1\circ \iota_{Y_1}$. Make analogous definitions for $h_{X_2}$ and $h_{Y_2}$. Define $g_X^+:X_{\DM_1}\to X_{\DM_2}$ to be the $X$ component of $\Delta \Gamma g$, which is the restriction of $\DM(\psi_g)$ to $X_{\DM_1}$. Define $g_Y^+$ analogously. Consider the diagram in Figure \ref{F:etaNat}.  The inner squares commute by definition of $g$, and that the outer squares commute can be deduced from the commutativity of the diagram in Figure \ref{F:Xcom}, the commutativity of whose right square follows from the commutativity of the diagram in Figure \ref{F:polHom}. 

Now, $\Delta \Gamma g\circ \eta_{G_1}$ is the polarity morphism $(g^+_X\circ h_{X_1},g_P,g^+_Y\circ h_{Y_1})$, and $\eta_{G_2}\circ g$ is the polarity morphism $(h_{X_2}\circ g_X, g_P, h_{Y_2}\circ g_Y)$, and these are equal by the commutativity of the diagram in Figure \ref{F:etaNat}. Thus $\eta$ is a natural transformation. 

Let $G=\Pe$ be a Galois polarity extending $P$. We will show that $\eta_G$ has the appropriate universal property (see e.g. \cite[Theorem 2.3.6]{Lein14}). Let $d:Q\to D$ be a $\Delta_1$-completion, and let $h=(h_X, h_P, h_Y): G \to \Delta (d)$ be a polarity morphism. We must find a map $g:\Gamma(G)\to d$ such that $\Delta g\circ\eta_G = h$, and show that $g$ is unique with this property. 

By Lemma \ref{L:GalEmb}, the map $\eta_{\Delta (d)}$ is an isomorphism, as $\Delta (d)$ is a complete Galois polarity, and is the unique such polarity isomorphism that is the identity on $Q$. 
Consider the diagram in Figure \ref{F:hCom}. The upper triangle commutes as $\eta$ is a natural transformation (see Figure \ref{F:hCom2}).  By Lemma \ref{L:delIsom} there is an isomorphism between $\Gamma\Delta (d)$ and $d$ (as extensions of $Q$), which we define to be $\tau = (\id_Q,t)$, and illustrate in Figure \ref{F:tau}. As functors preserve isomorphisms, $\Delta\tau$ is an isomorphism. Since $\Delta\tau$ is also the identity on $Q$, it must be the inverse of $\eta_{\Delta(d)}$. Thus, noting Figure \ref{F:hCom}, $\tau\circ \Gamma h:\Gamma(G)\to d$ has the property that 
\[\Delta (\tau\circ \Gamma h)\circ \eta_G = \Delta \tau\circ \Delta \Gamma h\circ \eta_G = \eta_{\Delta (d)}^{-1}\circ \Delta \Gamma h\circ\eta_G = h.\]
We must show that $\tau\circ \Gamma h$ is unique with this property, so let $f:\Gamma(G)\to d$ be a $\del$ morphism, and suppose $f$ is defined by $(f_1,f_2)$ for some $f_1:P\to Q$ and $f_2: \DM(X\uplus Y)\to D$. Suppose also that $\Delta f\circ\eta_G = h$. Then $\Delta  f = (f_2|_{X_\DM},f_1,f_2|_{Y_\DM})$, and, continuing to identify $X_\DM$ and $Y_\DM$ with the corresponding subsets of $\DM(X\uplus Y)$, we have 
\[\Delta  f\circ \eta_G = (f_2|_{X_\DM}\circ e\circ\iota_X, f_1\circ \id_P, f_2|_{Y_\DM}\circ e\circ \iota_Y).\]

Suppose $\tau\circ \Gamma h$ is defined by the maps $t_1:P\to Q$ and $t_2:\DM(X\uplus Y)\to D$. Then 
\[\Delta (\tau\circ \Gamma h)\circ \eta_G = (t_2|_{X_\DM}\circ e\circ\iota_X, t_1\circ \id_P, t_2|_{Y_\DM}\circ e\circ \iota_Y).\]
Since we have assumed that both $\Delta  f\circ \eta_G$ and $\Delta (\tau\circ \Gamma h)\circ \eta_G $ are equal to $h$, we must have $f_1= t_1 = h_P$, and, moreover, $f_2$ must agree with $t_2$ on $e[X\uplus Y]$. But $f_2|_{e[X\uplus Y]}$ is cut-stable, as it is the restriction of a complete lattice homomorphism between MacNeille completions, and so it extends uniquely to a complete lattice homomorphism (by \cite[Theorem 3.1]{Ern91a}). Thus $f_2 = t_2$, and so $f = \tau\circ \Gamma h$ as required. It follows that $\Gamma\dashv \Delta $ as claimed.            
\end{enumerate}
\end{proof}

\begin{figure}[htbp]
\begin{minipage}[b]{0.49\textwidth}
\[\xymatrix{ 
P_1\ar[d]_{h_P}\ar[r]^{\gamma_1} & X_1\uplus Y_1\ar[r]^{e_1}\ar[d]^{\psi_h} & \DM(X_1\uplus Y_1)\ar[d]^{\DM(\psi_h)} \\
P_2\ar[r]\ar[r]_{\gamma_2} & X_2\uplus Y_2\ar[r]_{e_2} & \DM(X_2\uplus Y_2)
}\]  
\caption{}
\label{F:Fmap}
  \end{minipage}
	 \hfill
  \begin{minipage}[b]{0.49\textwidth}
  \[\xymatrix{ 
P_1\ar[r]^{d_1}\ar[d]_{f} & D_1\ar[d]^g \\
P_2\ar[r]_{d_2} & D_2
}\]  
\caption{}
\label{F:Dmap}
  \end{minipage}
\end{figure}

\begin{figure}[htbp]
\begin{minipage}[b]{0.30\textwidth}
\[\xymatrix{ 
G_1\ar[r]^{\eta_{G_1}}\ar[d]_g & \Delta \Gamma(G_1)\ar[d]^{\Delta \Gamma g} \\
G_2\ar[r]_{\eta_{G_2}} & \Delta \Gamma(G_2)
}\] 
\caption{}
\label{F:etaSquare}
\end{minipage}
\hfill
\begin{minipage}[b]{0.69\textwidth}
\[\xymatrix{ 
X_{\DM_1}\ar[d]_{g^+_X} & X_1\ar[d]\ar[l]_{h_{X_1}}\ar[d]^{g_X} & \ar[l]_{e_{X_1}}P_1\ar[r]^{e_{Y_1}}\ar[d]^{g_P} & Y_1\ar[r]^{h_{Y_1}}\ar[d]^{g_Y} & Y_{\DM_1}\ar[d]^{g^+_Y}\\
X_{\DM_2} & X_2\ar[l]^{h_{X_2}}  & \ar[l]^{e_{X_2}}P_2\ar[r]_{e_{Y_2}} & Y_2\ar[r]_{h_{Y_2}}  & Y_{\DM_2}
}\] 
\caption{}
\label{F:etaNat}
\end{minipage}
\end{figure}

\begin{figure}[htbp]
\begin{minipage}[b]{0.49\textwidth}
\[\xymatrix{ 
\DM(X_1\uplus Y_1)\ar[d]_{\DM(\psi_g)} & \ar[l]_{e_1}X_1\uplus Y_1\ar[d]^{\psi_g} & \ar[l]_{\iota_{X_1}}X_1\ar[d]^{g_X} \\
\DM(X_2\uplus Y_2) & \ar[l]^{e_2}X_2\uplus Y_2 & \ar[l]^{\iota_{X_2}}X_2 
}\] 
\caption{}
\label{F:Xcom}
\end{minipage}
\hfill
\begin{minipage}[b]{0.49\textwidth}
\[\xymatrix{ 
G\ar[d]_h\ar[rr]^{\eta_G} & & \Delta \Gamma(G)\ar[d]^{\Delta \Gamma h} \\
\Delta (d)\ar[rr]_{\eta_{\Delta (d)}} & & \Delta \Gamma\Delta (d)
}\] 
\caption{}
\label{F:hCom2}
\end{minipage}
\end{figure}

\begin{figure}[htbp]
\begin{minipage}[b]{0.49\textwidth}
\[\xymatrix{ 
G\ar[ddrr]_h\ar[drr]^{\eta_{\Delta (d)}\circ h}\ar[rr]^{\eta_G} & & \Delta \Gamma(G)\ar[d]^{\Delta \Gamma h} \\
& & \Delta \Gamma\Delta (d) \\
& & \Delta (d)\ar[u]_{\eta_{\Delta(d)}(\cong)}
}\] 
\caption{}
\label{F:hCom}
\end{minipage}
\hfill
\begin{minipage}[b]{0.49\textwidth}
\[\xymatrix{ 
Q\ar[rrr]^{\Gamma\Delta (d) \phantom{c}= \phantom{c}e_d\circ\gamma_d\phantom{ccc}}\ar[drrr]_d & & & \DM(X_D\uplus Y_D)\ar[d]^{t(\cong)} \\
& & & D
}\] 
\caption{}
\label{F:tau}
\end{minipage}
\end{figure}

The components of counit of the adjunction between $\Gamma$ and $\Delta $ are provided by the isomorphisms produced in Lemma \ref{L:delIsom}. Thus the subcategory, $\textrm{Fix}(\Gamma\Delta )$,  of $\del$ is just $\del$ itself. The canonical categorical equivalence between $\textrm{Fix}(\Delta \Gamma)$ and $\textrm{Fix}(\Gamma\Delta )$ produces a categorical version of the correspondence in Theorem \ref{T:delta}.

We end with a universal property for Galois polarities whose relation is the minimal $\R_l$.

\begin{prop}\label{P:univ}
Let $G=(e_X,e_Y,\R_l)$ be a Galois polarity extending $P$, let $Q$ be a poset, and let $f:X\to Q$ and $g:Y\to Q$ be order preserving maps such that $f\circ e_X = g\circ e_Y$. Let $\preceq$ be the unique 3-preorder for $G$.  Then the following are equivalent:
\begin{enumerate}
\item $y\preceq x \ra g(y)\leq f(x) $ for all $x\in X$ and for all $y\in Y$.
\item There is a unique order preserving map $u:X\uplus Y\to Q$ such that the diagram in Figure \ref{F:univ} commutes.
\end{enumerate}
\end{prop}
\begin{proof}
Suppose (1) holds. We define $u':X\cup_\preceq Y \to Q$ by 
\[u'(z) = \begin{cases} f(z) \text{ if } z\in X \\ 
g(z) \text{ if } z\in Y\end{cases}\]
We show that $u'$ is order preserving with respect to the preorder $\preceq$ and the order on $Q$. Let $z_1\preceq z_2\in X\cup_{\preceq} Y$. If $z_1$ and $z_2$ are both in $X$, or both in $Y$, then that $u'(z_1)\leq u'(z_2)$ follows immediately from the definition of $u$ and the fact that both $f$ and $g$ are order preserving. If $z_1= x\in X$, and $z_2=y\in Y$, then there is $p\in e_X^{-1}(x^\uparrow)\cap e_Y^{-1}(y^\downarrow)$, and so  $f(x)\leq g(y)$ by the assumption that $f\circ e_X = g\circ e_Y$. If $z_1= y\in Y$ and $z_2 = x\in X$, then that $g(y)\leq f(x)$ is true by (1), and so $u'(y)\leq u'(x)$ as required. Define $u$ by $u(\iota_X(x)) = f(x)$ and $u(\iota_Y(y))= g(y)$. Then $u$ is well defined and order preserving by the monotonicity of $u'$, and that $u$ is unique with these properties is automatic from the required commutativity of the diagram.

Conversely, suppose (2) holds. Then 
\[y\preceq x \implies \iota_Y(y)\leq \iota_X(x)\implies u\circ \iota_Y(y)\leq u\circ \iota_X(x)\implies g(y)\leq f(x)\]
as required.  
\end{proof}

Proposition \ref{P:univ} says that, if $e_X$ and $e_Y$ are fixed meet- and join-extensions respectively, the pair of maps $(\iota'_X, \iota'_Y)$ arising from $(e_X,e_Y,\R_l)$ is initial in the category whose objects are pairs of order preserving maps $(f:X\to Q, g:Y\to Q)$ such that $f\circ e_X = g\circ e_Y$ and $y\preceq x \ra g(y)\leq f(x)$, and whose maps are commuting triangles as in Figure \ref{F:univMorph} (here $h$ is order preserving, and commutativity means $f_2 = h\circ f_1$ and $g_2 = h\circ g_1$). In particular this category contains all $(\iota_X,\iota_Y)$ arising from Galois polarities $\Pe$ based on $e_X$ and $e_Y$.

\begin{figure}[!tbbp]
  \begin{minipage}[b]{0.49\textwidth}
  \[\xymatrix{ 
P\ar[r]^{e_Y}\ar[d]_{e_X} & Y\ar[d]^{\iota_Y}\ar@/^/[ddr]^g\\
X\ar[r]_{\iota_X}\ar@/_/[drr]_f & X\uplus Y\ar[dr]^u \\
& & Q
}\] 
\caption{}
\label{F:univ}
  \end{minipage}
  \hfill
  \begin{minipage}[b]{0.49\textwidth}
  \[\xymatrix{ 
(X,Y)\ar[r]^{(f_1,g_1)}\ar[dr]_{(f_2,g_2)} & Q_1\ar[d]^h\\
& Q_2
}\] 
\caption{}
\label{F:univMorph}
  \end{minipage}
\end{figure}

\section*{Acknowledgement}
The author would like to thank the anonymous referees for their thorough reports, and for suggestions which considerably improved the clarity of the paper.

%\bibliography{../../../../references}{}
\bibliographystyle{abbrv}

\end{document}